\documentclass[11pt,twoside]{article}

\usepackage[ansinew]{inputenc} 
\usepackage{amsmath,amsfonts,amssymb,amsthm}
\usepackage{pstricks,pst-node,pst-coil,pst-plot,pstricks-add}
\usepackage{geometry,epsfig}
\usepackage{bbm}
\usepackage{cite}
\usepackage{graphicx}
\usepackage{xcolor}

\usepackage{bbm}
\usepackage{cite}
\usepackage{paralist}
\usepackage{cleveref}
\usepackage{enumitem}
\usepackage{emptypage}

\DeclareGraphicsExtensions{.pdf}

\usepackage{mathrsfs} 

\usepackage[T1]{fontenc}

\newcommand{\field}[1]{\mathbb{#1}}
\newcommand{\N}{\field{N}}
\newcommand{\R}{\field{R}}
\newcommand{\C}{\field{C}}
\newcommand{\Z}{\field{Z}}

\newcommand{\HH}{\mathscr H}

\newcommand{\LL}{\mathscr L}
\newcommand{\FF}{\mathcal F}
\newcommand{\hh}{\mathfrak h}

\newcommand{\Estar}{\mu} 
\newcommand{\Hrel}{H_{\mathrm{rel}}}

\newcommand{\eps}{\varepsilon}
\newcommand{\ph}{\varphi}
\newcommand{\const}{\mathrm{const}}
\newcommand{\ran}{\mathrm{Ran}}
\renewcommand{\ker}{\mathrm{Ker}}

\newcommand{\linhull}{\mathrm{span}}

\newcommand{\restricted}{|\grave{}\,}

\newcommand{\norm}[1]{\mbox{$\left\| #1 \right\|$}}           
\newcommand{\sprod}[2]{\langle #1,#2 \rangle}        
\newcommand{\ket}[1]{\left| #1 \right\rangle}               
\newcommand{\bra}[1]{\left\langle #1 \right|}               

\newcommand{\rst}{\! \upharpoonright \!} 

\newtheorem{thm}{Theorem}[section]
\newtheorem{prop}[thm]{Proposition}
\newtheorem{cor}[thm]{Corollary}
\newtheorem{lemma}[thm]{Lemma}

\title{\textbf{Spectral Theory of the Fermi Polaron}}

\author{M.~Griesemer\footnote{marcel.griesemer@mathematik.uni-stuttgart.de}\ and U.~Linden\footnote{ulrich.linden@mathematik.uni-stuttgart.de}\\  
\small Fachbereich Mathematik, Universit\"at Stuttgart, D-70569 Stuttgart, Germany}  
\date{}

\begin{document}
\maketitle

\begin{abstract} 
The Fermi polaron refers to a system of free fermions interacting with an impurity particle by means of two-body contact forces. Motivated by the physicists' approach to this system, the present article describes a general mathematical framework for defining many-body Hamiltonians with two-body contact interactions by means of a renormali\-zation procedure. In the case of the Fermi polaron the well-known TMS Hamiltonians are shown to emerge. For the Fermi polaron in a box  $[0,L]^2\subset \R^2$ a novel variational principle, established within the general framework, links the low-lying eigenvalues of the system to the zero-modes of a Birman-Schwinger type operator. It allows us to show, e.g., that the \emph{polaron}- and \emph{molecule} energies, computed in the physical literature, are indeed upper bounds to the ground state energy of the system.
\end{abstract}

\section{Introduction}


The Fermi polaron is a popular model in theoretical physics describing a gas of ideal fermions in contact with an impurity particle, the interaction being an attractive point interaction. This model describes, e.g., a sea of fermionic spin-up atoms in contact with a spin down atom of the same species, which is the case of extreme imbalance in the spin population in a gas of fermionic atoms. Formally, the Hamiltonian of the Fermi polaron reads
\begin{equation} \label{Formal_Expression}
   - \frac{1}{M} \Delta_y - \sum_{i = 1}^N \Delta_{x_i} - g \sum_{i = 1}^N \delta(x_i - y),
\end{equation}
where $y$ and $x_i$ denote the positions of the impurity and the fermions, respectively, $M>0$ is the mass of the impurity, $ \delta(x_i - y)$ denotes a Dirac-$\delta$-potential and $g$ plays the role of a coupling constant. Experiments on ultra cold gases of fermionic atoms with imbalanced spin population \cite{Nature} have triggered the interest in the analysis of this model (see e.g. \cite{Bruun_Review,Enss,BruunMassignan,ProkoSvistu,Parish,ParishLevinsen}). The  physicists are interested in the form of the ground state as a function of the coupling strength and in the possibility to observe a so called BEC-BCS crossover. The debate on these issues for two-dimensional systems is part of the motivation for the present work.

In the physics literature the starting point in the analysis of the Fermi polaron is a second quantized version of  \eqref{Formal_Expression} with an ultraviolett cutoff $\Lambda$ imposed on the high momenta involved in the fermion-impurity interaction \cite{CombescotGiraud,Parish}. This cutoff is eventually sent to infinity, while the two-body (one fermion and the impurity) binding energy $E_B$ is kept fixed. The regularized Hamiltonian $H_\Lambda$ is a function of the cutoff $\Lambda$ without an obvious limit as $\Lambda\to \infty$. Nevertheless, the existence of a non-trivial model emerging in this limit is taken for granted by the physicists, and the attention is focused on the form of the ground state as a function of $E_B$. There are two distinct families of variational states, the \emph{polaron} and the \emph{molecule} states, that are considered good approximations to the ground state at weak and strong coupling, respectively \cite{Chevy,ChevyMora,PDZ}.  The stationarity of the energy of these variational states is expressed in terms of non-linear, implicit equations for the Lagrange multiplier associated with the normalization condition. Results of numerical solutions of these equations are taken as evidence that the ground state in the weak coupling regime is well approximated by a polaron state whereas in the strong coupling regime it is better approximated by a molecule state \cite{Parish,ParishLevinsen}.

The present paper is inspired by the work described above and by previous mathematical work on contact interactions \cite{Figari, DR}. Following the approach of \cite{DR}, we develop a general mathematical framework for studying the spectrum of many-particle systems with two-body contact interactions. This framework has the same structure as the singular perturbation theory developed by Posilicano \cite{Pos2001, Pos2008,CFP2018}, but it has a different starting point, hence other hypotheses, and a different focus. It is taylormade for the Fermi polaron in a box $[0,L]^2\subset \R^2$, for which it allows us to establish existence of a self-adjoint Hamiltonian $H_N$ in terms of the resolvent limit, as $\Lambda\to\infty$, of the regularised, second quantised Hamiltonian $H_\Lambda$ restricted to the space $\HH_N$ of $N$ fermions and the impurity. A generalized Birman-Schwinger operator $\phi(z)$, $z\in \C$ acting on a smaller space $\tilde\HH_{N-1}$ plays an important role in this proof and in the subsequent analysis of the spectrum of $H_N$. Let $H_{0,N}$ denote the Hamiltonian of $N$ fermions and the impurity without any interactions. For $z\in \rho(H_{0,N})$ we show there is a bounded operator $B_z\in \LL(\HH_N,\tilde\HH_{N-1})$ such that
\begin{equation}\label{Krein}
     (H_N-z)^{-1} = (H_{0,N}-z)^{-1} + B_{\bar z}^{*}\phi(z)^{-1} B_z
\end{equation}
where $B_z$ is an isomorphism from $\ker(H_N-z)$ to $\ker(\phi(z))$. This allows us to prove for $E<\min\sigma(H_{0,N})$ that
\begin{equation} \label{Birman_Schwinger_Principle}
   \mu_{\ell}(H_N) \leq  E \quad \Leftrightarrow \quad \mu_{\ell} (\phi(E)) \leq 0,
\end{equation}
where $\mu_{\ell}(\cdot)$ denotes the $\ell$th eigenvalue counted from below with multiplicities. Equivalence \eqref{Birman_Schwinger_Principle}, which also holds with strict inequalities, implies that $H_N-E$ and $\phi(E)$ have the same number of negative eigenvalues, which is analog to the familiar Birman-Schwinger principle for the negative eigenvalues of Schr\"odinger operators. In view of \eqref{Birman_Schwinger_Principle} with $\ell=1$, any solution $E$ to an equation $\sprod{w}{\phi(E)w}=0$, with $w\in \tilde\HH_{N-1}\backslash\{0\}$, is an upper bound to the ground state energy $\mu_1(H_N)$ of $H_N$. The vector $w$  is arbitrary and subject to optimisation. Mapping the variational states from the physical literature to the smaller space $\tilde\HH_{N-1}$ we construct analogs of the polaron and molecule states for the Birman-Schwinger operator $\phi(z)$. This mapping reduces the set of coefficients significantly and simplifies the variational computations compared to the work by the physicists. We moreover reproduce the results from the physics literature for the energy of the polaron and the molecule, and by doing so, we prove that these expressions are upper bounds to the ground state energy of the Hamiltonian $H_N$ as defined above.  From the point of view of applications this is the main result of the present paper.  Further applications of our new variational principle are published elsewhere:  in \cite{Diss-Linden} it is shown that the molecule energy is indeed lower than the polaron energy in the limit of large $|E_B|$, in \cite{GriesemerLinden1} stability of the two-dimensional fermi-polaron is established, and in \cite{LindenMitrouskas2018} it is shown the polaron energy correctly describes the ground state energy in the high density limit if the impurity mass is infinite.

	In the mathematical literature many-particle Hamiltonians with contact interactions are usually described in terms of TMS-Hamiltonians, named after Ter-Martirosyan and Skornyakov, which are defined in terms of boundary conditions at the collision planes $x_i=x_j$ for the free Hamiltonian \cite{Five_Italians,Five_Italians_Simplified,Minlos_N+1,Five_Italians_study3,MichPfei2016,MoserSeiringer,MichOtto2017,MichOtto2018,MoSe2018,MoSe2018b}. The focus in most of these papers is on the questions of  self-adjointness and stability, where stability refers to the dependence of the ground state energy on the number $N$ of fermions and the mass $M$ of the impurity.  Qualitative aspects of the spectrum are analyzed in \cite{BMO2018}, and one-dimensional systems of three particles with point-interactions are shown to be resolvent limits of scaled Schr\"odinger operators in \cite{Basti2018}.  - All the above work on many-particle TMS-Hamiltonians, with the exception of \cite{Basti2018}, is based on a construction of the Hamiltonian in terms of semi-bounded closed form described for the first time in \cite{Figari}.  This quadratic form can be seen as the $\Gamma-$limit of approximating forms with a UV-cutoff in the relative particle momenta \cite{Figari}, but it can also be written down directly, see e.g.~\cite{Five_Italians_Simplified}. An alternative approach for defining two-body-delta-interactions was presented in \cite{DR} for the system of bosons in 2d. In this new approach the Hamiltonian is defined as a strong resolvent limit of  UV-regularized, second quantized Hamiltonians. For $N=2$ it was also shown that this new approach leads to TMS Hamiltonians but for  $N>2$ this question was left open.  While $\Gamma-$convergence is closely related to strong resolvent convergence \cite{DalMaso}, and a UV regularization played an important role both in \cite{Figari} and \cite{DR}, it is far from obvious whether the two construction lead to the same Hamiltonian. It is one of the main objectives of the present paper to clarify this point for the system of our main concern, the Fermi-polaron in two dimensions.

All the questions addressed in this paper for the 2d Fermi polaron in a square box with periodic boundary conditions can equally be studied with other boundary conditions, other traps, and in three-dimensional space. This can be done with the tools developed in the present paper, as the abstract part, Sections  \ref{sec:BS-operator} - \ref{Sect:Variational_Principle}, is independent of such model characteristics.  Our choice of square boxes with periodic boundary conditions follows the physical literature and it is motivated by our search for the meaning of the polaron and molecule equations.   
The exclusion of three dimensions avoids the so called Thomas effect, a spectral phenomenon that occurs for small values of the impurity mass $M$ \cite{Five_Italians}. Proving absence of a Thomas effect for large $M$ amounts to serious technical difficulties with the verification of the hypotheses of \Cref{Convergence_Theorem} (see Section~\ref{Sect:Regularization_Schemes}). But there are no principle obstacles, and  based on \cite{MoserSeiringer} we conjecture that the statement of \Cref{Regularization_Theorem} holds in three dimensions as well provided that $M>0.36$. 

This paper is organized as follows: Section \ref{Sect:Fermi_Polaron_Box} first describes the regularized Hamiltonian of the Fermi polaron in second quantized form. Then, the free parameter $E_B<0$ in this Hamiltonian is shown to agree with the ground state energy in the two-body subspace of the sector of vanishing total momentum. Sections  \ref{sec:BS-operator} - \ref{Sect:Variational_Principle} are devoted to Hamiltonians that are given in terms of  resolvent limits of sequences of semi-bounded self-adjoint operators of the general form $H_n=H_0-g_n A_n^{*}A_n$, $n\in \N$.  The equivalence \eqref{Birman_Schwinger_Principle} is established in this general setting. The two-dimensional Fermi polaron fits into this general framework, as we show in Sections~ \ref{Sect:Regularization_Schemes} and \ref{sec:all-space}. Section~\ref{sec:all-space} shows, in addition,  that vectors in the domain of the Hamiltonian are characterized by the TMS condition.
In Section \ref{Sect:Pol_and_Mol_Equation} we derive the polaron and the molecule equations from the physical literature and we use \eqref{Birman_Schwinger_Principle} to prove that the solutions to these equations are upper bounds to the ground state energy of $H_N$. Finally, Section~\ref{sec:two-species} explains how the more general class of systems consisting of $N_1+N_2$ particles from two species of fermions fits into the abstract framework of Sections~\ref{sec:BS-operator} - \ref{Sect:Variational_Principle}.


\section{The regularized Hamiltonian in second quantization}
 \label{Sect:Fermi_Polaron_Box}

This paper is mainly concerned with a system of $N$ identical fermions and a single impurity in a two-dimensional box $\Omega = [0,L]^2$ with periodic boundary conditions. The Hilbert space of this system is given by
\begin{equation} \label{Representation_H_1}
   \HH_N := L^2(\Omega) \otimes \bigwedge\nolimits^{\!N} L^2(\Omega).
\end{equation}
Since we work in second quantization, we consider $\HH_N$ as a subspace of $\FF\otimes\FF$, where $\FF$ is the antisymmetric Fock space over $L^2(\Omega)$, and we define the regularized Hamiltonian on $\FF\otimes\FF$. To this end we need an ONB of $L^2(\Omega)$. In view of the periodic boundary conditions, a suitable orthonormal basis of $L^2(\Omega)$ is given by the functions 
$$
   \varphi_k(x) := \frac{e^{ikx}}{L} \quad \text{for} \: k \in \kappa \Z^2,\quad  \kappa := \frac{2\pi}{L}.
$$
The corresponding fermionic annihilation and creation operators are denoted by $a_k$ and  $a^{*}_k$ in the case of the fermions, and by $b_k$ and  $b^{*}_k$ in the case of the impurity particle. 

The second quantization of  \eqref{Formal_Expression} expressed in terms of $a_k, a^{*}_k$ and  $b_k, b^{*}_k$ is not a well-defined operator. In the physics literature this problem is solved by dropping terms in the interaction part with momentum $k \in \kappa \Z^2$ of magnitude larger than some cutoff $\Lambda$. The precise form of the cutoffs is immaterial, as we will see in this paper. We work with a general class of cutoffs in terms of two functions $\alpha,\beta : \kappa \Z^2 \to [0,1]$ subject to the condition that 
\begin{equation} \label{First_Condition_Regularization}
   C(\alpha,\beta) := \sup_{q \in \kappa \Z^2} \sum_k |\alpha(k) \beta(q-k)|^2 < \infty.
\end{equation}
Here and in the following all sums run over the elements of the momentum lattice $\kappa \Z^2$. Condition \eqref{First_Condition_Regularization} will allow us to show in Lemma~\ref{Regularized_Interaction} that the following expressions give well-defined operators on $\HH_N$, for any $N\in\N$. We set 
\begin{equation} \label{H_alpha_beta}
   H_{\alpha,\beta}:= H_0 - g_{\alpha,\beta} W_{\alpha,\beta}
\end{equation}
where
\begin{align}
       H_0 &:= \sum_k k^2 (a^*_k a_k + \tfrac{1}{M} b^*_k b_k), \label{Definition_H_0}\\
       W_{\alpha, \beta} &:= \sum_{k,l,q} \alpha(k) \alpha(l) \beta(q-k) \beta(q-l) \: a_k^* \: b_{q-k}^* \: b_{q-l} \: a_l \label{W-ab}
\end{align}
and 
\begin{equation} \label{CouplingConstant}
   g_{\alpha,\beta}^{-1} = \sum_{k} \frac{\alpha(k)^2 \beta(-k)^2}{(1+\frac{1}{M}) k^2 - E_B}.
\end{equation}
The number $E_B < 0$ is a free parameter that determines the coupling strength. By our choice of  $g_{\alpha,\beta}$, $E_B$  is the ground state energy of $H_{\alpha,\beta}$ in the two-body sector $\HH_{N=1}$ with total momentum zero. This is proved below.

It is essential for our analysis, that $W_{\alpha,\beta}=V_{\alpha,\beta}^{*}V_{\alpha,\beta}$ on $\HH_N$ with a suitable operator $V_{\alpha,\beta}:\HH_N \to \tilde\HH_{N-1}$. This operator on $\FF\otimes \FF$ is given by 
\begin{equation}\label{V-ab}
   V_{\alpha, \beta} := \sum_{k,q} \alpha(k) \beta(q-k) \: m_q^* \: b_{q-k} \: a_k,
\end{equation}
where $m_q$ is another notation for $b_q$ stressing the distinct role of the particle, called \emph{angel} in \cite{DR}, that is created by $V_{\alpha, \beta}$. Similarly, $\tilde\HH_{N-1}$ is another notation for $\HH_{N-1}$ reminding us that the additional particle is the angel and not the impurity. The motivation for distinguishing the angel from the impurity becomes fully clear Section~\ref{sec:two-species}, where we discuss the generalization to systems of $N_1+N_2$ fermions. Then $m_q\neq b_q$. 

The desired identity $W_{\alpha, \beta} = V_{\alpha, \beta}^{*} V_{\alpha, \beta}$ follows from the fact that $m_p m_q^* = \delta_{p,q}$ on the vacuum sector.

\begin{lemma} \label{Regularized_Interaction}
Let $\alpha,\beta$ be real-valued functions on $\kappa\Z^2$ satisfying \eqref{First_Condition_Regularization}. Then, for each $N\in\N$, the series \eqref{W-ab} and \eqref{V-ab} define bounded sesquilinear forms on $\HH_N\times\HH_N$ and  $\tilde\HH_{N-1}\times\HH_N$, respectively. The corresponding operators $W_{\alpha,\beta}\in \LL(\HH_N)$ and  $V_{\alpha,\beta}\in \LL(\HH_N,\tilde\HH_{N-1})$ obey,
\begin{itemize}
\item[(1)] $\| V_{\alpha, \beta}\| \leq \sqrt{N} \cdot C(\alpha,\beta)^{1/2}$,
\item[(2)] $W_{\alpha,\beta} =V_{\alpha,\beta}^{*}V_{\alpha,\beta}.$
\end{itemize}
\end{lemma}

\begin{proof}
The series defining $\sprod{w}{ V_{\alpha, \beta}\psi}$ for $w \in \tilde\HH_{N-1}$ and $\psi \in \HH_N$ is absolutely convergent. Indeed,
\begin{eqnarray*}
      \lefteqn{\sum_{k,q} \alpha(k) \beta(q-k) |\sprod{m_q w}{b_{q-k} a_k\psi}|}\\
                   & &\leq \sum_{k,q} |\alpha(k) \beta(q-k)| \cdot \norm{m_q w} \cdot \norm{b_{q-k} a_k \psi} \\
                    &&\leq \sum_q \left( \sum_{k} |\alpha(k) \beta(q-k)|^2 \right)^{1/2} \cdot \left( \sum_k \norm{b_{q-k} a_k \psi}^2 \right)^{1/2} \cdot \norm{m_q w} \\
                    &&\leq C(\alpha,\beta)^{1/2} \cdot \sum_{q} \left( \sum_k \norm{b_{q-k} a_k \psi}^2 \right)^{\! 1/2} \norm{m_q w} \\
                    &&\leq C(\alpha,\beta)^{1/2} \cdot \left( \sum_{k,q} \norm{b_{q-k} a_k \psi}^2 \right)^{\! 1/2} \cdot \left( \sum_q \norm{m_q w}^2 \right)^{\! 1/2} \\
                    &&= C(\alpha,\beta)^{1/2} \cdot \sqrt{N} \cdot \norm{\psi} \cdot \norm{w}.
\end{eqnarray*}
In a similar way one shows that the series defining $\sprod{\ph}{W_{\alpha,\beta}\psi}$ for $\ph,\psi\in \HH_N$ is absolutely convergent. This allows us, in particular, to exchange summands in this series at will.  Using that $m_q m_p^* = \delta_{q,p}$ on the vacuum sector we obtain,
\begin{align*}
    \sprod{\ph}{W_{\alpha,\beta}\psi} &= \sum_{k,\ell}\sum_q \alpha(k)\alpha(\ell)\beta(q-k)\beta(q-\ell)\sprod{b_{q-k}a_k\ph}{b_{q-\ell}a_\ell\psi}\\
    &= \sum_{k,\ell}\sum_{q}\sum_{p} \alpha(k)\alpha(\ell)\beta(q-k)\beta(p-\ell)\sprod{m_q^{*}b_{q-k}a_k\ph}{m_p^{*}b_{p-\ell}a_\ell\psi}\\
    &= \sum_{k,q} \alpha(k)\beta(q-k) \sprod{m_q^{*}b_{q-k}a_k\ph}{V_{\alpha,\beta}\psi}\\
    & = \sprod{V_{\alpha,\beta}\ph}{V_{\alpha,\beta}\psi},
\end{align*}
which completes the proof.
\end{proof}

It remains to explain the choice \eqref{CouplingConstant} for  $g_{\alpha,\beta}$. As pointed out above, the regularized Hamiltonian $H_{\alpha, \beta}$ is of the form $H_0 - g V_{\alpha, \beta}^* V_{\alpha, \beta}$ and hence  \Cref{General_Connection_H_Phi} from the next section applies. By this proposition, $\lambda<0$ is an eigenvalue of $H_{\alpha, \beta}$ on $\HH_{N=1}$ if and only if the operator $\phi(\lambda) := g_{\alpha, \beta}^{-1} - V_{\alpha, \beta} (H_0 - \lambda)^{-1} V_{\alpha,\beta}^*$ has a non-trivial kernel in the target space $\tilde\HH_{N=0}=L^2(\Omega)$ of $V_{\alpha,\beta}\rst \HH_{N=1}$. Upon normal ordering with the help of the pull-through formulas \eqref{Pull_Through_Formula}, we obtain 
\begin{align*}
   \phi(\lambda)\rst \tilde\HH_{N=0} &= \sum_q f(\lambda,q) m_q^* m_q,\\
    f(\lambda,q) &:=  g_{\alpha,\beta}^{-1} - \sum_k \frac{\alpha(k)^2 \beta(q-k)^2}{\frac{1}{M}(q-k)^2 + k^2 - \lambda}.
\end{align*}
The kernel of this operator is non-trivial if and only if $f(\lambda,q)=0$ for some $q$. Then  $m_q^* \ket{\text{vac}} \in \ker(\phi(\lambda))$ and, by  \Cref{General_Connection_H_Phi},
$$
   R_0(\lambda) V_{\alpha,\beta}^* m_q^* \ket{\text{vac}} = \sum_k \alpha(k) \beta(q-k) R_0(\lambda) a_k^* b_{q-k}^* \ket{\text{vac}}
$$
is the corresponding eigenvector of $H_{\alpha, \beta}\rst \HH_{N=1}$. This is a state of total momentum $q$. From the monotonicity of $\lambda\mapsto f(\lambda,q)$ on $\R_{-}$ it is clear, that $f(\lambda,q)=0$ has at most one solution $\lambda<0$, and, by \eqref{CouplingConstant}, $f(E_B,0)=0$. It follows that $E_B$ is the only negative eigenvalue and hence the ground state of $H_{\alpha, \beta}\rst \HH_{N=1}$ in the sector of momentum $q=0$.


\section{Schur complements and the generalized Birman-Schwinger operator}
\label{sec:BS-operator}

In this section we give a general discussion of operators of the form
\begin{equation}\label{General_Form_Operator}
   H = H_0 - g A^* A,
\end{equation}
where $H_0$ is a positive operator on a Hilbert space $\HH$, $g > 0$ is a positive coupling constant, and $A \in \LL(\HH,\tilde\HH)$, where $\tilde\HH$ is another Hilbert space, possibly different from $\HH$. The regularized Hamiltonians defined in the previous section and in Section~\ref{sec:all-space} are of this form. Another instructive and important example of operators of the type \eqref{General_Form_Operator} is given at the end of this section.

For $z \in \rho(H_0)$ we define an operator $\phi(z): \tilde\HH \to \tilde\HH$ by
\begin{equation}\label{Birman_Schwinger_Operator}
   \phi(z) := g^{-1} - A R_0(z) A^*, 
\end{equation}
where
$$
   R_0(z) := (H_0 - z)^{-1}.
$$
The operator \eqref{Birman_Schwinger_Operator} will be called the (generalized) Birman-Schwinger operator of $H$ at the point $z$.
This is justified by the analogy with the Birman-Schwinger operator for Schr\"odinger operators, which would correspond to $A R_0(z) A^*$, and by the following proposition.

\begin{prop}\label{General_Connection_H_Phi}
\text{}
\begin{enumerate}
 \item[(a)] Let $z \in \rho(H_0)$. Then,
$$
   z \in \rho(H) \: \Leftrightarrow \: 0 \in \rho(\phi(z)),
$$
and the resolvents $R(z) := (H-z)^{-1}$ and $\phi(z)^{-1}$ are connected by the equations
\begin{align}
 R(z) &= R_0(z) + R_0(z) A^* \phi(z)^{-1} A R_0(z), \label{Resolvent}\\
 \phi(z)^{-1} &= g + g^2 A R(z) A^*. \label{Inverse_of_Phi}
\end{align}
 \item[(b)] $z\in \rho(H_0)$ is an eigenvalue of $H$ if and only if $0$ is an eigenvalue of $\phi(z)$. Moreover,
\begin{align*}
   A: \ker(H-z) \to \ker(\phi(z)) \\
   R_0(z) A^*: \ker(\phi(z)) \to \ker(H-z)
\end{align*}
are isomorphisms.
\end{enumerate}
\end{prop}

\begin{proof}
To prove (a) we define the block operator
$$
   \tilde{H}(z) = \left( \begin{array}{cc} H_0 - z & A^* \\ A & g^{-1} \end{array} \right) : D(H_0) \oplus \tilde\HH \to \HH \oplus \tilde\HH
$$
The following identities are straightforward to verify:
\begin{align}
   \tilde{H}(z) &= \left( \begin{array}{cc} 1 & gA^* \\ 0 & 1 \end{array} \right) \cdot \left( \begin{array}{cc} H - z & 0 \\ 0 & g^{-1} \end{array} \right) \cdot \left( \begin{array}{cc} 1 & 0 \\ gA & 1 \end{array} \right), \label{Matrix_Decomposition_1}\\
   \tilde{H}(z) &= \left( \begin{array}{cc} 1 & 0 \\ A R_0(z) & 1 \end{array} \right) \cdot \left( \begin{array}{cc} H_0 - z & 0 \\ 0 & \phi(z) \end{array} \right) \cdot \left( \begin{array}{cc} 1 & R_0(z)A^* \\ 0 & 1 \end{array} \right). \label{Matrix_Decomposition_2}
\end{align}
We see that $H-z$ is the first Schur complement of $\tilde{H}(z)$ while $\phi(z)$ is the second Schur complement (cf.~\cite{Tretter}). The triangular block operators with identities on the diagonal have bounded inverses, which are obtained by changing the sign of the off-diagonal terms. From \eqref{Matrix_Decomposition_1} and \eqref{Matrix_Decomposition_2} we can read off
$$
   0 \in \rho(\tilde{H}(z)) \Leftrightarrow z \in \rho(H)
$$
and
$$
   0 \in \rho(\tilde{H}(z)) \Leftrightarrow 0 \in \rho(\phi(z)),
$$
respectively. We combine both statements and obtain
$$
   z \in \rho(H) \Leftrightarrow 0 \in \rho(\phi(z)).
$$
In this case we can invert both sides of expressions \eqref{Matrix_Decomposition_1} and \eqref{Matrix_Decomposition_2}.
\begin{align*}
   \tilde{H}(z)^{-1} &= \left( \begin{array}{cc} R(z) & -g R(z) A^* \\ -g A R(z) & g + g^2 A R(z) A^* \end{array} \right) \\
   \tilde{H}(z)^{-1} &= \left( \begin{array}{cc} R_0(z) + R_0(z) A^* \phi(z)^{-1} A R_0(z) & -g A^* \phi(z)^{-1} \\ -\phi(z)^{-1} A R_0(z) & \phi(z)^{-1} \end{array} \right) 
\end{align*}
A comparison of the two equations yields \eqref{Resolvent} and \eqref{Inverse_of_Phi} and the proof of (a) is complete.

From \eqref{Matrix_Decomposition_1}, \eqref{Matrix_Decomposition_2}, from the invertibility of the triangular block operators and the invertibility of $H_0-z$ we obtain the equivalences
$$
   \ker(H-z) \neq \{ 0 \} \quad \Leftrightarrow \quad \ker(\tilde{H}(z)) \neq \{ 0 \} \quad \Leftrightarrow \quad \ker(\phi(z)) \neq \{ 0 \},
$$
and
\begin{align*}
   & \left( \begin{array}{c} \psi \\ w \end{array} \right) \in \ker(\tilde{H}(z)) \subseteq D(H_0) \oplus \tilde\HH \\
   \Leftrightarrow \qquad & \psi \in \ker(H-z) \: \land \: w + g A \psi = 0 \\
   \Leftrightarrow \qquad & w \in \ker(\phi(z)) \: \land \: \psi + R_0(z) A^* w = 0.
\end{align*}
This proves (b).
\end{proof}

\noindent
\textbf{Example.} Let $H_0$ be a positive operator on some Hilbert space $\HH$, let $\eta\in \HH\backslash\{0\}$ and let $H=H_0-g|\eta\rangle\langle\eta|$. Then $H$ is the special case of \eqref{General_Form_Operator}, where $A \in \LL(\HH, \C)$ is given by $A \psi = \sprod{\eta}{\psi}$ and hence  $A^* \psi \in \LL(\C, \HH)$ with $A^* c = c \cdot \eta$ for $c \in \C$. It follows that 
$$
       \phi(z) = g^{-1} -\sprod{\eta}{(H_0-z)^{-1}\eta},
$$
and that $\lambda\in \rho(H_0)$ is an eigenvalue of $H$ if and only if $\phi(\lambda)=0$. It is then a straightforward computation to verify that 
$$
     (H_0-\lambda)^{-1}\eta
$$
is an eigenvector of $H$ associated with $\lambda$. By \Cref{General_Connection_H_Phi}, the resolvent of $H$ for $z\in \rho(H)\cap\rho(H_0)$ is given by 
$$
      R(z) = R_0(z) + \phi(z)^{-1} R_0(z)|\eta\rangle\langle\eta| R_0(z).
$$

\section{The Hamiltonian in the strong resolvent limit}
\label{Sect:Strong_Resolvent_Limit}

Now we consider sequences $(H_n)_{n \in \N}$ of operators of the form \eqref{General_Form_Operator} and we establish sufficient conditions for the strong resolvent convergence of such sequences. \Cref{Convergence_Theorem}, below, is the key tool for our construction of the Hamiltonian of the Fermi polaron in a two-dimensional box (\Cref{Sect:Regularization_Schemes}). As a preparation we first prove:

\begin{lemma}\label{Convergence_Inverse}
 Let $T_n,T:D\subset \HH\to \HH$ be essentially self-adjoint operators and suppose that  $T_n \geq c>0$ for all $n \in \N$ and some $c\in\R$. If $T_n \psi \to T \psi$ as $n \to \infty$  for all $\psi \in D$, then $\overline{T} \geq c$ and
$$
   \overline{T}_n^{-1} \to \overline{T}^{-1} \qquad (n \to \infty)
$$
in the strong operator topology.
\end{lemma}

\begin{proof}
From $T_n\geq c$ and $T_n\psi \to T\psi$ it follows that $T\geq c$. Passing to the closures we see that
\begin{equation}\label{CI-1}
      \overline{T}_n\geq c\quad \text{and}\quad \overline{T}\geq c.
\end{equation}
Since $\overline{T}$ is self-adjoint it follows that $0\in \rho(\overline{T})$ and hence that $\ran\overline{T}=\HH$, which implies that $\ran T$ is dense.  Since, by \eqref{CI-1}, $\overline{T}_n^{-1}$ is uniformly bounded, it suffices to prove the desired convergence on a dense subspace such as $\ran T\subset\HH$. For $\psi = T\ph$ with $\ph\in D$ we have 
\begin{align*}
  \overline{T}_n^{-1}\psi - \overline{T}^{-1}\psi &= \overline{T}_n^{-1}(T\ph) -\ph
  = \overline{T}_n^{-1}(T\ph -T_n\ph) \to 0, \quad (n\to\infty).
\end{align*}
\end{proof}

\begin{thm}\label{Convergence_Theorem}
 Let $(A_n)_{n \in \N}$ be a sequence in $\LL(\HH, \tilde\HH)$, let $(g_n)_{n \in \N}$ be a sequence of positive numbers, and let $$\phi_n(z) := g_n^{-1} - A_n R_0(z) A_n^*\qquad\text{for}\ z\in \rho(H_0).$$ Suppose there exists a number $\Estar<0$ such that following hypotheses are satisfied.
\begin{enumerate}[font=\normalfont, label={(\textrm{\alph*})}]
 \item The limit $B_\Estar := \lim_{n\to\infty} A_n R_0(\Estar)$ exists in $\LL(\HH, \tilde\HH)$.
 \item There is a dense subspace $D \subseteq \tilde\HH$ and an essentially self-adjoint operator $\phi(\Estar): D \to \tilde\HH$ such that for $\psi\in D$, $\phi_n(\Estar)\psi \to \phi(\Estar)\psi$ as $n \to \infty$.
 \item There is a positive number $c > 0$ such that $\phi_n(\Estar) \geq c$ for all $n \in \N$.
\end{enumerate}
Then, the sequence $H_n := H_0 - g_n A_n^* A_n$ has a limit $H: D(H) \to \HH$ in the strong resolvent sense.
The operator $H$ is self-adjoint, $H>\Estar$, and 
\begin{equation}\label{Resolvent_Limit_Operator}
   (H - \Estar)^{-1} = R_0(\Estar) + B_\Estar^* \phi(\Estar)^{-1} B_\Estar.
\end{equation}
\end{thm}

\noindent
\textbf{Remarks.}
\begin{enumerate}
\item In applications of this theorem the numbers $g_n$ are chosen in such a way that Hypothesis (b) is satisfied.
\item By the resolvent identity for $R_0(z)$, Hypothesis (a) implies that for all $z \in \rho(H_0)$
\begin{equation} \label{RE_Rz}
   A_n R_0(z) \to B_z:= B_\Estar + (z - \Estar) B_\Estar R_0(z)\qquad (n\to\infty).
\end{equation}
Moreover, by (a), the operator $A:D(H_0)\subset\HH\to \tilde\HH$ defined by 
$$
        A\ph := \lim_{n\to\infty}A_n\ph,\qquad \ph\in D(H_0)
$$ 
exists and $B_z = AR_0(z)$ for all $z\in \rho(H_0)$. We are interested in the case where $A$ is an unbounded operator. Boundedness of $A$ implies $\ran B_z^{*}\subset D(H_0)$ which, by \Cref{Notin_Lemma}, is not true in the context of Section~\ref{Sect:Regularization_Schemes}.

\item If Hypotheses (a) and (b) are satisfied, then, by the previous remark, for all $z \in \rho(H_0)$ and all $w\in D$ we have $\phi_n(z)w\to \phi(z)w$ as $n\to\infty$, where 
\begin{equation} \label{PhiE_Phiz_kurz}
   \phi(z) = \phi(\Estar) + (\Estar-z) B_z B_{\Estar}^*.
\end{equation}
Since $\phi(z) - \phi(\Estar)$ is a bounded operator, it follows that $\phi(z)$ is closable and that the domain of the closure is independent of $z$.
Moreover, $\phi(\lambda)$ is essentially self-adjoint for all $\lambda < 0$. In the following, the closure of $\phi(z)|_D$ is denoted by $\phi(z)$ as well, and its domain is denoted by $D(\phi)$.

\item By the monotonicity of the resolvent $\tau\mapsto R_0(\tau)$,  $\phi_n(\tau)\geq \phi_n(\Estar)$ if $\tau \leq \Estar$. Hence, Hypotheses (a)-(c) are satisfied for all $\tau \leq \Estar$. 
\end{enumerate}

\begin{proof}[Proof of \Cref{Convergence_Theorem}]
By definition, $\phi_n(\Estar)$ is a bounded self-adjoint operator, which, by Assumptions (b), (c) satisfies the hypotheses of \Cref{Convergence_Inverse}. It follows that $0$ belongs to the resolvent set of $\phi(\Estar)$ and that 
\begin{equation}\label{RR1}
   \phi_n(\Estar)^{-1} \to \phi(\Estar)^{-1} \qquad (n \to \infty)
\end{equation}
strongly. By Proposition \ref{General_Connection_H_Phi}, the number $\Estar$ belongs to the resolvent set of $H_n$ and
\begin{equation} \label{Regularized_Resolvent}
   (H_n - \Estar)^{-1} = R_0(\Estar) + R_0(\Estar) A_n^* \phi_n(\Estar)^{-1} A_n R_0(\Estar).
\end{equation}
In view of the strong convergence of \eqref{RR1}, Assumption (a) and \eqref{Regularized_Resolvent} imply that 
\begin{equation}\label{RR2}
   (H_n - \Estar)^{-1} \to R(\Estar) := R_0(\Estar) + B_{\Estar}^* \phi(\Estar)^{-1} B_{\Estar} \qquad (n\to\infty)
\end{equation}
strongly. The operator $R(\Estar)$ is bounded, self-adjoint, and strictly positive, that is $R(\Estar)>0$, because $R_0(\Estar)>0$ and $\phi(\Estar)^{-1}>0$. It follows that $R(\Estar)\HH$ is dense and that the operator
$$
      H:= R(\Estar)^{-1} + \Estar
$$
is self-adjoint with domain  $R(\Estar)\HH$, $\Estar$ is in the resolvent set, and $H>\Estar$ by the positivity of $R(\Estar)^{-1}$. By Theorem VIII.19 of \cite{RS1}, it remains to prove the strong resolvent convergence for some $z\in \C\backslash\R$.

The strong convergence \eqref{RR2} implies that $\sup_n\norm{(H_n - \Estar)^{-1}} < 1/\eps$ for $\eps$ small enough. All $z\in B(\Estar,\eps)$ belong to the resolvent set of $H_n$ and
$$
   (H_n  - z)^{-1} = \sum_{k=0}^\infty (z-\Estar)^k (H_n - \Estar)^{-k-1}.
$$
From this equation and from the strong convergence $(H_n - \Estar)^{-1}\to (H-\Estar)^{-1}$ we see that for $z\in B(\Estar,\eps)$,
$$
   (H_n-z)^{-1} \to R(z) := \sum_{k=0}^\infty (z-\Estar)^k (H-\Estar)^{-k-1}\qquad (n\to\infty)
$$
strongly. It is easy to check that $R(z) = (H-z)^{-1}$ and the proof is complete. 
\end{proof}

The above remarks in combination with  \Cref{Convergence_Theorem}  imply the following corollary.  See also \Cref{Discrete_Spectrum_H}, below, and the remark thereafter.


\begin{cor} \label{resolvent-on-R}
For all $\tau\leq \Estar$ and $B_\tau$ defined by \eqref{RE_Rz},
\begin{equation} \label{eq:Resolvent_of_H}
     (H - \tau)^{-1} = (H_0 - \tau)^{-1} + B_{\tau}^{*}\phi(\tau)^{-1} B_\tau.
\end{equation}
\end{cor}

\noindent
\textbf{Remark:}  An abstract resolvent identity identical to \eqref{eq:Resolvent_of_H} previously appeared in the singular perturbation theory of Posilicano \cite{Pos2001,Pos2008,CFP2018}. In this theory, the given objects are a self-adjoint operator $H_0$, an $H_0$-bounded operator $A$, and an operator-valued function $\phi(z)$, $z\in \rho(H_0)$, with properties like \eqref{PhiE_Phiz_kurz}, where $B_z=A(H_0-z)^{-1}$. Assuming that $\ker A$ is dense in $\HH$, or, at least that $\ran B_{\bar z}^{*}\cap D(H_0) = \{0\}$, it is shown in \cite{Pos2001} that the r.h.s of \eqref{eq:Resolvent_of_H} is the resolvent of a self-adjoint operator $H$ extending $H_0\restricted\ker A$. There is more to say about \eqref{eq:Resolvent_of_H} in the remark following \Cref{Discrete_Spectrum_H}.

\medskip
The representation of $D(H)$ given in the following proposition is inspired by the work on so called TMS Hamiltonians \cite{Figari, Five_Italians, MichOtto2017}.  Recall that $D(\phi)$ denotes the domain of $\phi(z)$, which is independent of $z \in \rho(H_0)$.

\begin{prop}\label{Explicit_Characterization_of_H}
A vector $\ph\in \HH$ belongs to $D(H)$ if and only if there exists a vector $w_\ph\in D(\phi)$ such that for some (and hence all) $z\in \rho(H_0)$, 
\begin{equation}\label{D-conditions}
     \ph - B_{\bar{z}}^* w_\ph \in D(H_0)\quad \text{and}\quad A(\ph - B_{\bar{z}}^*  w_\ph) = \phi(z) w_\ph.
\end{equation}
If this is the case, then 
\begin{equation}\label{Action_of_H} 
   (H - z) \ph = (H_0 - z)(\ph - B_{\bar{z}}^*  w_\ph). 
\end{equation}
\end{prop}

\noindent
\textbf{Remarks.} 
\begin{enumerate}
\item While \eqref{Action_of_H}  appears to give an explicit expression for $H\ph$, it does not because it depends on the vector $w_\ph$ whose dependence on $\ph$ is not explicit. 
\item Equation \eqref{D-conditions} can be seen as an abstract, operator theoretic version of the so called TMS boundary condition \cite{Figari, Five_Italians, MichOtto2017}. In Section \ref{sec:all-space} we show how this condition reduces to the usual TMS condition in the case of the Fermi-polaron in $\R^2$.
\item In the application to the Fermi polaron we know that $B_{\overline{z}}^* w\not\in D(H_0)$ unless $w=0$ (see \Cref{Notin_Lemma}). This and $ \ph - B_{\overline{z}}^* w_\ph \in D(H_0)$ will imply the uniqueness of $w_\ph$.
\end{enumerate}

\begin{proof}
Pick $\tau \leq \mu$ with $\mu$ given by \Cref{Convergence_Theorem}. We will have occasion to use the identity 
\begin{equation}\label{eq:Rtau*}
    B_{\tau}^{*} = B_{\bar z}^{*} + (\tau-z) R_0(z)B_{\tau}^{*}, 
\end{equation}
which follows from \eqref{RE_Rz} and which holds for all $z\in \rho(H_0)$.

Assume that $\ph \in D(H)$, define $v_\ph:= (H-\tau)\ph \in \HH$ and  $w_\ph := \phi(\tau)^{-1} B_\tau v_\ph\in D(\phi)$. Then $\ph=(H-\tau)^{-1}v_\ph$ and, by rearranging \eqref{eq:Resolvent_of_H}, we find
\begin{equation} \label{Relation_v_phi}
   \ph - B_{\tau}^* w_\ph = R_0(\tau) v_\ph  \in D(H_0).
\end{equation}
This implies that 
\begin{equation}\label{D-for-tau}
   A(\ph - B_\tau^* w_\ph) =  A R_0(\tau) v_\ph  =  B_\tau v_\ph = \phi(\tau) w_\ph
\end{equation}
and hence \eqref{D-conditions} is proved for $z=\tau$. From \eqref{eq:Rtau*} and \eqref{Relation_v_phi} it is obvious that $ \ph - B_{\bar z}^* w_\ph \in D(H_0)$ for all $z\in \rho(H_0)$ and we claim that \eqref{D-for-tau} extends to all $z\in \rho(H_0)$ as well.  Indeed, using \eqref{PhiE_Phiz_kurz} and \eqref{eq:Rtau*}  we see that
\begin{equation} \label{z-to-tau}
    A(\ph - B_{\bar{z}}^*w_\ph) - \phi(z) w_\ph  =  A(\ph - B_{\tau}^* w_\ph) -  \phi(\tau) w_\ph,
\end{equation}
which completes the proof of \eqref{D-conditions}.

Now let $\ph \in \HH$ and assume there exists a vector $w_\ph \in D(\phi)$ such that \eqref{D-conditions} holds for some $z\in \rho(H_0)$. Then  \eqref{D-conditions} holds for all
$z\in \rho(H_0)$ by \eqref{z-to-tau} and by the arguments preceding it.  Define $v_\ph := (H_0 - \tau)(\ph - B_{\tau}^* w_\ph)$. Then, by \eqref{D-conditions} for $z=\tau$,
\begin{align*}
   R_0(\tau) v_\ph &= \ph - B_{\tau}^* w_\ph\\ 
   &= \ph - B_{\tau}^* \phi(\tau)^{-1} \phi(\tau) w_\ph\\ 
   &= \ph - B_{\tau}^* \phi(\tau)^{-1} A(\ph - B_\tau^* w_\ph) \\
   &= \ph - B_{\tau}^* \phi(\tau)^{-1} A R_0(\tau) v_\ph = \ph - B_\tau^* \phi(\tau)^{-1} B_\tau v_\ph.
\end{align*}
By \eqref{eq:Resolvent_of_H}, this implies $\ph = (H-\tau)^{-1} v_\ph \in D(H)$ and $(H - \tau) \ph = v_\ph = (H_0 - \tau)(\ph - B_\tau^* w_\ph)$. Using this last equation and \eqref{eq:Rtau*} we conclude that
\begin{align*}
     (H - z)\ph & = (H - \tau)\ph + (\tau-z)\ph\\
      &= (H_0 - \tau)(\ph - B_{\tau}^* w_\ph) + (\tau-z)\ph\\
      &= (H_0 - z)(\ph - B_{\bar{z}}^* w_\ph)
\end{align*}
which completes the proof.
\end{proof}

\begin{cor}\label{lower-bound} 
For all $E<0$,
$$
       \phi(E) \geq 0\quad\Rightarrow\quad H\geq E.
$$
\end{cor}

\begin{proof}
For all $E<0$ and all $\ph\in D(H)$, by \eqref{Action_of_H}, \eqref{D-conditions}, and $R_{E}(H_0-E) = A$ on $D(H_0)$,
\begin{align*}
    \sprod{\ph}{(H-E)\ph} = \sprod{(\ph - B_{E}^* w_\ph)}{(H_0-E)(\ph - B_{E}^* w_\ph)}  + \sprod{w_\ph}{\phi(E)w_\ph}.
\end{align*}
Since $H_0-E\geq 0$ this equation proves the corollary.
\end{proof}


\begin{cor}\label{Kernel_and_Eigenvalues} 
With the notations and assumptions of \Cref{Convergence_Theorem} for all $z\in \rho(H_0)$, 
\item[(a)] $B_{\bar{z}}^{*}\ker\phi(z) = \ker(H-z)$.
\item[(b)] If $\ker B_{\bar z}^* = \{0\}$, then $z$ is an eigenvalue of $H$ if an only if $0$ is an eigenvalue of $\phi(z)$.
\end{cor}

In \Cref{Discrete_Spectrum_H} below this corollary will be strengthened in the case where $H_0$ has a compact resolvent. 

\begin{proof}
Part (b) follows from (a). To prove (a) let $w\in \ker\phi(z)$ and  define $\ph := B_{\bar z}^{*}w $. Then $\ph - B_{\bar{z}}^* w = 0\in D(H_0)$ and $A(\ph - R_z^* w) =0=\phi(z)w$. Hence, by \Cref{Explicit_Characterization_of_H}, $\ph\in D(H)$ and 
$$
   (H - z)\ph = (H_0 - z)(\ph - B_{\bar{z}}^{*}w) = 0.
$$
This proves that $B_{\bar{z}}^{*}\ker\phi(z)\subset \ker(H-z)$. Now let $\ph\in \ker(H-z)$. Then,  by \Cref{Explicit_Characterization_of_H},
\begin{equation}\label{Action_of_H_on_Kernel}
   \ph - B_{\bar z}^* w_\ph = (H_0 - z)^{-1}(H-z)\ph =0 
\end{equation}
for some $w_\ph\in D(\phi)$, and, using \eqref{D-conditions},
$$
   \phi(z) w_\ph = A(\ph - B_{\bar z}^* w_\ph) = 0.
$$
This proves that $B_{\bar{z}}^{*}\ker\phi(z)\supset \ker(H-z)$. 
\end{proof}

We now illustrate the results of this section in the easy case of a single quantum particle subject to a $\delta$-potential sitting at the origin and confined to a two-dimensional box $\Omega = [0,L]^2$ with periodic boundary conditions. To this end we consider the sequence of operators
$$
   H_n := -\Delta - g_n \ket{\eta_n} \bra{\eta_n}
$$
with $\eta_n = \sum_{k^2 \leq n} \ph_k$, $\ph_k$ as in \Cref{Sect:Fermi_Polaron_Box}, and $g_n$ defined by 
$$
   g_n^{-1} = \sum_{k^2 \leq n} \frac{1}{k^2 - E_B}.
$$
Here $E_B<0$ is a free parameter of the system, which, as the following shows, agrees with the ground state energy of $H_n$. Notice that $\sprod{\eta_n}{\ph}\to L\cdot\ph(0)$ as $n\to\infty$ for smooth $\ph\in L^2(\Omega)$ satisfying periodic boundary conditions. 

The operators $H_n$ are of the type considered in the example of \Cref{sec:BS-operator}. That is, $H_n = -\Delta - g_n A_n^* A_n$ where $A_n \in \LL(L^2(\Omega), \C)$ is given by $A_n \psi = \sprod{\eta_n}{\psi}$, $A_n^* \in \LL(\C, L^2(\Omega))$ acts as $A_n^* c = c \cdot \eta_n$, and
$$
   \phi_n(z) = g_n^{-1} - \sprod{\eta_n}{(-\Delta - z)^{-1} \eta_n} = \sum_{k^2 \leq n} \left( \frac{1}{k^2 - E_B} - \frac{1}{k^2 - z} \right).
$$
It is easy to see that Conditions (a), (b) and (c) of \Cref{Convergence_Theorem} are met. In particular, $\eta(z) := \lim_{n \to \infty} (-\Delta - z)^{-1} \eta_n  = \sum_k(k^2-z)^{-1}\ph_k$ and $\phi(z)=\lim_{n\to\infty}\phi_n(z)$ exist. Notice that $\phi_n(E_B)=0$ and that, for any $\mu<E_B$, $\phi_n(\mu)\geq c>0$ for all $n$. By \Cref{Convergence_Theorem}, we conclude that $H_n \to H$ as $n \to \infty$ in the strong resolvent sense (actually in the norm resolvent sense) for a self-adjoint operator $H \geq E_B$. Moreover, by \Cref{Discrete_Spectrum_H} below,
$$
   (H-z)^{-1} = (-\Delta - z)^{-1} + \phi(z)^{-1} |\eta(z)\rangle\langle\eta(\bar{z})|
$$
for all $z\in \rho(H_0)\cap \rho(H)$ and $E_B$ is an eigenvalue of $H$ with eigenvector $\eta(E_B)$.

A similar discussion of external $\delta$-potentials as limits of rank one perturbations can be found in \cite{SingularPerturbations}. For a comprehensive discussion of $\delta$-potentials the reader is referred to \cite{SolvableModels}.

\section{The variational principle}

\label{Sect:Variational_Principle}

Given $E<\inf\sigma(H_0)$ let $\mu_\ell(H)$ denote the $\ell$-th eigenvalue from below, counting multiplicities, of the semi-bounded, self-adjoint operator $H$. Recall that 
\begin{equation}\label{min-max}
     \mu_\ell(H) := \min_{\substack{M\subset D(\phi)\\ \dim(M)=\ell}} \left(\max_{\substack{u \in M, \norm{u}=1}} \sprod{u}{H u}\right).
\end{equation}
In this section we prove that $\mu_\ell(H)$, if it is below $\inf\sigma(H_0)$, is the unique zero of the function $\tau\mapsto \mu_\ell(\phi(\tau))$ on $(-\infty,\min\sigma(H_0))$. To this end we assume:

\begin{itemize}
\item[(H1)] The resolvent of $H_0$ is compact
\item[(H2)] $\ker B_z^* =\{0\}$ for all $z \in\rho(H_0)$
\end{itemize}

\begin{prop}\label{Discrete_Spectrum_Phi}\label{Discrete_Spectrum_H}
In addition to the hypotheses of \Cref{Convergence_Theorem}, assume $(H1)$ and $(H2)$. Let $c\in \R$ be defined by Hypothesis (c) of \Cref{Convergence_Theorem}. Then 
\begin{itemize}
\item[(a)] The operator $H$ has purely discrete spectrum. The operator $\phi(z)$, for $z \in \rho(H_0)$, has purely discrete spectrum in $\C\backslash [c,\infty)$.
\item[(b)] A number $z\in \rho(H_0)$ is an eigenvalue of $H$ if and only if $0\in \sigma(\phi(z))$. On $\rho(H)\cap\rho(H_0)$ the map $z\mapsto \phi(z)^{-1}$ is analytic and 
\begin{equation}\label{res-H}
      (H - z)^{-1} = (H_0 - z)^{-1} + B_{\bar z}^{*} \phi(z)^{-1} B_{z}.
\end{equation}
\end{itemize}
\end{prop}

\noindent
\textbf{Remark.} The second part of statement (b) is true without the assumptions  $(H1)$ and $(H2)$. This follows from \cite{CFP2018}; see Theorem 2.19, Remark 2.20, and the explicit formula for $\phi(z)^{-1}$ in that paper. Assuming $(H1)$ and $(H2)$, which are satisfied for the system considered in Sections \ref{Sect:Regularization_Schemes} and \ref{Sect:Pol_and_Mol_Equation}, we give a short, independent proof.

\begin{proof}
From the compactness of $R_0(z)$ it follows that $B_z := \lim_{n\to\infty} A_n R_0(z)$ is compact. Hence $ (H - \Estar)^{-1} = R_0(\Estar) + B_{\Estar}^* \phi(\Estar)^{-1} B_{\Estar}$, with $\Estar < 0$ defined by \Cref{Convergence_Theorem}, is compact.
This implies that the spectrum of $H$ is discrete.  

We now prove (b).  From  \eqref{PhiE_Phiz_kurz} we know that 
\begin{align}\label{phi-K}
   \phi(z) &= \phi(\Estar) - (z-\Estar)B_zB_{\Estar}^{*} = (1-K(z))\phi(\Estar)  
\end{align}
where $K(z) = (z-\Estar)B_zB_{\Estar}^{*}\phi(\Estar)^{-1}$ is compact and analytic as a function of $z\in\rho(H_0)$. The analyticity of $B_z$ follows from  \eqref{RE_Rz}. Since $1-K(z)$ is invertible for $z=\Estar$, it follows from the analytic Fredholm theorem, that $(1-K(z))^{-1}$ exists and is analytic for $z\in \rho(H_0)$ except for poles at the points where $1-K(z)$, and hence $\phi(z)$, have a non-trivial kernels. By \Cref{Kernel_and_Eigenvalues} these are exactly the eigenvalues of $H$ in $\rho(H_0)$. This shows that both sides of  \eqref{res-H} are analytic on $\rho(H)\cap\rho(H_0)$. Since both sides agree if $z\in (-\infty,\Estar)$, by \Cref{resolvent-on-R}, it follows that \eqref{res-H} holds for all  $z\in \rho(H)\cap\rho(H_0)$.

The second part of (a) follows from the compactness of $\phi(z)-\phi(\Estar)$ by Theorem XIII.14 in \cite{RS4}. Instead of verifying the hypotheses of that theorem, we prefer to give a direct proof along the lines of the proof of (b). Let $z\in \rho(H_0)$ be fixed. Then, for $\lambda\in \C\backslash [c,\infty)$ with $c \in\R$ defined by \Cref{Convergence_Theorem}, 
$$
   \phi(z) - \lambda = (1 - K(z,\lambda)) (\phi(\Estar) - \lambda)
$$
where $K(z,\lambda)= (z-\Estar)B_zB_{\Estar}^{*} (\phi(\Estar) - \lambda)^{-1}$ is compact and analytic as a function of $\lambda$. From $\|K(z,\lambda)\|\to 0$ as $\lambda \to -\infty$, it is clear that $1- K(z,\lambda)$ is invertible for $\lambda$ negative and large. Hence by the meromorphic Fredholm theorem, Theorem XIII.13 of \cite{RS4}, $[1- K(z,\lambda)]^{-1}$ exists and is analytic as a function of $\lambda\in \C\backslash [c,\infty)$ except for poles, and the coefficients of the singular part of the Laurent series at these poles are finite rank operators. It follows that these poles belong to discrete eigenvalues of $\phi(z)$.
\end{proof}

\begin{prop}\label{Continuity_of_the_Minimum}
In addition to the hypotheses of \Cref{Convergence_Theorem}, assume $(H1)$ and $(H2)$. Then for $\tau,\tau'\in (-\infty,\min\sigma(H_0))$ and all $\ell\in\N$ the following is true:
\begin{itemize}
\item[(a)] $\ker\phi(\tau) \cap \ker\phi(\tau')=\{0\}$ if $\tau\neq\tau'$.
\item[(b)] The map $\tau\mapsto \mu_\ell(\phi(\tau))$ is continuous, and, if $\mu_\ell(\phi(\tau))<c$, it is strictly decreasing.
\end{itemize}
\end{prop}

\begin{proof}
By \eqref{PhiE_Phiz_kurz} with expression \eqref{RE_Rz} for $B_z$,  
\begin{equation}\label{d-phi}
    \phi(\tau) = \phi(\Estar) - (\tau- \Estar) B_{\Estar} B_{\Estar}^{*} - (\tau- \Estar)^2 B_{\Estar} R_0(\tau)B_{\Estar}^{*}.  
\end{equation}
This identity shows that $\tau\mapsto \phi(\tau) -\phi(\Estar)$ is continuous in operator norm, which, in view of \eqref{min-max} implies that $\tau\mapsto  \mu_\ell(\phi(\tau))$ is continuous. From \eqref{d-phi} and (H2) we see, moreover, that 
\begin{equation}\label{phi-monotone}
     \sprod{u}{\phi(\tau)u} < \sprod{u}{\phi(\Estar)u}
\end{equation}
for $\Estar<\tau<\inf\sigma(H_0)$ and all $u\in D(\phi)$ with $u\neq 0$. This proves (a) and it implies that  $\tau\mapsto  \mu_\ell(\phi(\tau))$ is strictly decreasing below $\inf\sigma(H_0)$. Indeed, let $M_{\ell-1}(\tau)\subset D(\phi)$ be an $(\ell-1)$-dimensional space  spanned by $\ell-1$ eigenvectors of $\phi(\tau)$ associated with the lowest $\ell-1$ eigenvalues, such that $\phi(\tau)\geq\mu_{\ell}(\phi(\tau))$ on $M_{\ell-1}(\tau)^{\perp}$. Let $M\subset D(\phi)$ be an $\ell$-dimensional space for which
$$
\mu_\ell(\phi(\Estar))  = \max_{\substack{u \in M, \norm{u}=1}} \sprod{u}{\phi(\Estar)u}.
$$
Then there exists a normalized vector $u\in M\cap M_{\ell-1}(\tau)^{\perp}$ and hence, by \eqref{phi-monotone}, 
$$
     \mu_\ell(\phi(\Estar)) \geq \sprod{u}{\phi(\Estar)u} >  \sprod{u}{\phi(\tau)u}\geq  \mu_{\ell}(\phi(\tau)).
$$
\end{proof}

\begin{thm}\label{Minimum_Expectation_Value}
For a real number $E < \min \sigma(H_0)$,
\begin{equation}
   \mu_\ell(H) = E \quad \Leftrightarrow \quad \mu_\ell(\phi(E)) = 0. \label{Variational_Principle_1}
\end{equation}
\end{thm}

\begin{proof}
Fix an arbitrary $E <  \min \sigma(H_0)$ such that at least one of the sets
\begin{align*}
 M &:= \{ k \in \mathbb{N} \: : \: \mu_k(H) = E \}, \\
 N &:= \{ k \in \mathbb{N} \: : \: \mu_k(\phi(E)) = 0 \}
\end{align*}
is non-empty. By \Cref{Kernel_and_Eigenvalues}, $|M| = |N|$. Hence, there are integers $a,b \geq 1$ and $r \geq 0$ such that
$$
 M = \{ a, ..., a+r \} \quad \text{and} \quad
 N = \{ b, ..., b+r \}.
$$
To prove the theorem it suffices to show that $a = b$.

Suppose $a < b$. Then, $\mu_{a}(\phi(E)) < 0$. From $\phi(\mu)\geq c>0$, see \Cref{Convergence_Theorem}, and  \Cref{Continuity_of_the_Minimum} it follows that there exist $E_1,\ldots,E_a$ such that
$$
   \mu < E_1 \leq E_2 \leq \ldots \leq E_a < E
$$
and $\mu_i(\phi(E_i)) = 0$ for $i \in \{ 1,\ldots,a \}$. By \Cref{Kernel_and_Eigenvalues}, each $E_i$ is an eigenvalue of $H$, and if there is a group of $m \geq 1$ which agree, that is $E_s = E_{s+1}=\ldots=E_{s+m-1}$ for some $s$, then $\dim \ker \phi(E_s) \geq m$. Hence $E_s$ is an eigenvalue of $H$ of multiplicity $m$ or higher. These arguments show that $H$ has at least $a$ eigenvalues below $E$, which is in contradiction to $a \in M$.

Suppose $b < a$. Then, $\mu_b(H) < E$ and hence
$$
   b \leq \sum_{t < E} \dim \ker(H-t).
$$
Of course, this sum has only finitely many non-zero summands. By \Cref{Kernel_and_Eigenvalues},
\begin{align*}
   \dim\ker(H-t) = \dim\ker(\phi(t)) = \left| N(t) \right|,
\end{align*}
where $N(t) := \left\{ j \in \N \:|\: \mu_j(\phi(t)) = 0 \right\}$. By \Cref{Continuity_of_the_Minimum} , $N(s) \cap N(t) = \emptyset$ if $s \neq t$. This implies
\begin{align*}
   \sum_{t < E} \left| N(t) \right| = \left| \bigcup_{t < E} N(t) \right|.
\end{align*}
We conclude that
$$
   b \leq \left| \bigcup_{t < E} N(t) \right|.
$$
Hence there exists $j\in N(t)$ for some $t < E$ with $j\geq b$. By definition of $N(t)$, $\mu_j(\phi(t)) = 0$, and by \Cref{Continuity_of_the_Minimum}, $\mu_b(\phi(E))\leq \mu_j(\phi(E)) < 0$. This is in contradiction to $b \in N$.
\end{proof}

From the proof of \Cref{Minimum_Expectation_Value} or from the statement of this theorem in combination with  \Cref{Continuity_of_the_Minimum} (b) we obtain:

\begin{cor}\label{Upper_Bound}
If $E < \min \sigma(H_0)$, then
$$
   \mu_\ell(H) \leq E \quad \Leftrightarrow \quad \mu_\ell(\phi(E)) \leq 0.
$$
In particular, if there is non-zero $w \in D(\phi)$ such that  
$$
   \sprod{w}{\phi(E) w} \leq 0,
$$
then $\min \sigma(H)\leq E$.
\end{cor}


\section{Regularization and strong resolvent convergence}
\label{Sect:Regularization_Schemes}

This section is devoted to the construction of the Hamiltonian $H:D(H)\subset\HH_N\to\HH_N$ of the Fermi polaron confined to a two-dimensional box with periodic boundary conditions. This operator is a realization of the formal expression \eqref{Formal_Expression}. We obtain $H$ as the limit in the strong resolvent sense of a sequence of regularized operators of the form \eqref{H_alpha_beta} with the help of \Cref{Convergence_Theorem}.  For three-dimensional systems the verification of Hypotheses (b) and (c) of \Cref{Convergence_Theorem} is much harder and requires a lower bound on $M$ \cite{Five_Italians}. On a technical level, in three dimensions the off-diagonal part of $\phi(z)$, see Equation~\eqref{explicit-phi} below, becomes unbounded, which is in contrast to \eqref{Phi_I_Bound}.

We recall from Section \ref{Sect:Fermi_Polaron_Box} that 
$$
C(\alpha,\beta) = \sup_{q \in \kappa \Z^2} \sum_k |\alpha(k) \beta(q-k)|^2.
$$

\begin{thm}\label{Regularization_Theorem}
Let $\alpha_n,\beta_n:\kappa \Z^2\to [0,1]$, $n\in\N$, be two sequences of functions on $\kappa \Z^2$ with $C(\alpha_n, \beta_n) < \infty$ for all $n \in \mathbb{N}$. Suppose that
\begin{enumerate}[label=(\Alph*), font=\normalfont]
 \item $\alpha_n(k), \beta_n(k) \to 1$ as $n \to \infty$ for all $k \in \kappa \Z^2$,
 \item $\gamma_n(q) \to 0$ as $n \to \infty$ for all $q \in \kappa \Z^2$ and $\sup_{n, q} |\gamma_n(q)| < \infty$, where
 \begin{equation} \label{gamma_Alternative_1}
  \gamma_n(q) = \sum_k |\alpha_n(k)| \cdot \frac{|\beta_n(-k) - \beta_n(q-k)|}{k^2 + q^2 + 1}
\end{equation}
or
\begin{equation} \label{gamma_Alternative_2}
  \gamma_n(q) = \sum_k |\beta_n(k)| \cdot \frac{|\alpha_n(-k) - \alpha_n(q-k)|}{k^2 + q^2 + 1}.
\end{equation}
\end{enumerate}
Let $H_0$ and $W_{\alpha_n, \beta_n}$ be defined as in Section~\ref{Sect:Fermi_Polaron_Box}, and let $g_n > 0$ be given by
\begin{equation}\label{Renormalization_Condition}
   g_n^{-1} = \sum_k \frac{\alpha_n(k)^2 \beta_n(-k)^2}{(1 + \frac{1}{M})k^2 - E_B},
\end{equation}
where $E_B < 0$ can be chosen arbitrarily as parameter of the system. Then, for all $N\in\N$ and $M>0$, the operators
$$
   H_n := H_0 - g_n W_{\alpha_n, \beta_n}
$$
converge to a self-adjoint operator $H$ in the strong resolvent sense as $n \to \infty$. The operator $H$ is bounded from below and does not depend on the choice of the sequences $(\alpha_n)_{n \in \N}$ and $(\beta_n)_{n \in \N}$.
\end{thm}

\noindent\textbf{Remarks.}
\begin{enumerate}

\item It is possible to choose $\alpha_n(k) \equiv 1$ and $\beta_n \in \ell^2(\kappa \Z^2; \R)$ with $\lim_{n\to\infty}\beta_n(k)=1$ for each $k$, or vice versa. Then $C(\alpha_n,\beta_n)=\|\beta_n\|^2<\infty$ and \eqref{gamma_Alternative_2} vanishes. Hence the Hypotheses of \Cref{Regularization_Theorem} are satisfied. 

\item The conditions of \Cref{Regularization_Theorem} are also satisfied if we choose $\alpha_n = \beta_n = \hat{\eta}_n$, where $\hat{\eta}_n$ denotes the characteristic function of the set $\{k\in\kappa\Z^2:|k|\leq n\}$. In this case (A) is obvious and  
$$
   \gamma_n(q) = \sum_{k: |k|\leq n<|k-q|} \frac{1}{k^2 + q^2 + 1}.
$$
It is not hard to show, see the proof of \Cref{Riemann}, that $\gamma_n(q)$ is uniformly bounded and that $\lim_{n\to\infty}\gamma_n(q)=0$.
\end{enumerate}

In the rest of \Cref{Sect:Regularization_Schemes}, we prove  \Cref{Regularization_Theorem}. Let $(\alpha_n)_{n \in \mathbb{N}}$ and $(\beta_n)_{n \in \mathbb{N}}$ be sequences satisfying the hypotheses of  \Cref{Regularization_Theorem}, and let
\begin{equation} \label{def:V_n}
   V_n := \sum_{k,q} \alpha_n(k) \beta_n(q-k) \: m_q^* \: b_{q-k} \: a_k.
\end{equation}
We recall from  \Cref{Sect:Fermi_Polaron_Box} that $W_n := W_{\alpha_n, \beta_n} = V_n^* V_n$. Therefore, the existence statement of \Cref{Regularization_Theorem} follows from \Cref{Convergence_Theorem} provided that we can verify Conditions (a), (b) and (c) of that theorem for $A_n=V_n$. The problems in veryfing (c) for polarons in three dimensions, which is not even possible for small $M$, are the main reason for restricting ourselves to two-dimensional systems.
The following lemma is devoted to the verification of Condition (a) of \Cref{Convergence_Theorem}. In its proof and in the proof of \Cref{Facts_about_Phi} the pull-through formulas
\begin{align}\label{Pull_Through_Formula}
   (H_0-z)^{-1} a_k^* = a_k^{*}(H_0 + k^2-z)^{-1},\qquad (H_0-z)^{-1} b_k^* = b_k^* \left(H_0+ \tfrac{1}{M} k^2-z\right)^{-1},
\end{align}
valid for $z\in \R_{-}\cup(\C\backslash\R)$, play an essential role.

\begin{lemma} \label{Norm_Convergence_V}
Let Hypothesis (A) of \Cref{Regularization_Theorem} be satisfied. Then, for all $z\in\rho(H_0)$, the sequence $V_n (H_0 - z)^{-1}\in\LL(\HH_N,\tilde{\HH}_{N-1})$ in the limit $n\to\infty$ converges in norm to a compact operator $B_z \in \LL(\HH_N,\tilde{\HH}_{N-1})$. Consequently, $(H_0 - \overline{z})^{-1} V_n^* \to B_z^*$, $V \psi := \lim_{n \to \infty} V_n \psi$ exists for all $\psi \in D(H_0)$, and $B_z = V (H_0 - z)^{-1}$. The operator $B_z$ is independent of the choice of the sequences $(\alpha_n)$ and $(\beta_n)$.
\end{lemma}

\begin{proof}
It suffices to prove the assertion for $z = E < 0$. Then, it will follow for all $z \in \rho(H_0)$ by the argument in Remark 2 following \Cref{Convergence_Theorem}. We will define $B_E$ by the sesquilinear form $b_E: \tilde{\HH} \times \HH \to \C$ given by
$$
   b_E(w, \psi) = \sum_{k,q} \sprod{m_q w}{b_{q-k} a_k (H_0-E)^{-1} \psi}.
$$
The following estimates show that this series is absolutely convergent and that its sum defines a bounded sesquilinear form.
For all $w\in  \tilde{\HH}_{N-1}$ and $\psi\in \HH_N$, by the pull-through formulas \eqref{Pull_Through_Formula} and the Cauchy-Schwarz-inequality for the $(k,q)$-sum,
\begin{align*}
   \lefteqn{\sum_{k,q} |\sprod{m_q w}{b_{q-k} a_k (H_0-E)^{-1} \psi}| }\\*
   &\leq \sum_{k,q} \norm{m_q w} \cdot \norm{(H_0 + \frac{1}{M} (q-k)^2 + k^2 - E)^{-1} b_{q-k} a_k \psi} \\
   &\leq \sum_{k,q} \norm{m_q w}(k^2 - E)^{-1} \cdot \|b_{q-k} a_k \psi\| \\
   &\leq \Bigg(\sum_{k,q} \|m_q w \|^2(k^2 - E)^{-2} \Bigg)^{1/2}  \Bigg( \sum_{k,q} \norm{b_{q-k} a_k \psi}^2 \Bigg)^{1/2} \\
   & = \sqrt{N}  \Bigg( \sum_k (k^2 - E)^{-2} \Bigg)^{1/2}  \norm{w} \cdot \norm{\psi},
\end{align*}
where $\sum_q b_q^{*}b_q=\sum_q m_q^{*}m_q=1$ and $\sum_q a_k^{*}a_k=N$ on $\HH_N$ was used in the last identity.
Hence, there is a bounded operator $R_E \in \LL(\HH_N, \tilde{\HH}_{N-1})$ such that
$$
   b_E(w,\psi) = \sprod{w}{B_E \psi}.
$$
We now show that $V_n(H_0 - E)^{-1} \to B_E$ as $n \to \infty$ by estimates similar to those above. By definition of $V_n$ and $B_E$, for $w \in \tilde{\HH}_{N-1}$ and $\psi \in \HH_N$,
\begin{align}
   &|\sprod{w}{(V_n(H_0 - E)^{-1} - B_E) \psi}|\nonumber \\*
   &\leq \sum_{k,q} \frac{|\alpha_n(k) \beta_n(q-k) - 1|}{\frac{1}{M}(q-k)^2 + k^2 - E}\cdot \|m_q w\| \cdot \|b_{q-k} a_k \psi\| \nonumber \\
   &\leq \Bigg( \sum_{k,q} \norm{b_{q-k} a_k \psi}^2 \Bigg)^{1/2}\left( \sum_q \norm{m_q w}^2 \sum_k \frac{|\alpha_n(k) \beta_n(q-k) - 1|^2}{(\frac{1}{M}(q-k)^2 + k^2 - E)^2} \right)^{1/2} \nonumber \\
   &\leq \sqrt{N} \cdot \norm{w} \cdot \norm{\psi} \cdot \left( \sup_q \sum_k \frac{|\alpha_n(k) \beta_n(q-k) - 1|^2}{(\frac{1}{M}(q-k)^2 + k^2 - E)^2} \right)^{1/2}. \label{Sup_Int_Constant}
\end{align}
It remains to show that \eqref{Sup_Int_Constant} vanishes in the limit $n\to\infty$. This will follow from that fact that, by (A) and dominated convergence, 
$$
   \sum_k \frac{|\alpha_n(k) \beta_n(q-k) - 1|^2}{(\frac{1}{M}(q-k)^2 + k^2 - E)^2} \to 0
$$
if either $|q|\to\infty$ or $n\to\infty$, where then the convergence is uniform in $n\in\N$ if $|q|\to\infty$. Hence, given $\eps>0$ there exists $C_\eps$ such that for $|q|>C_\eps$ and all $n\in \N$,
\begin{equation}\label{sup-eps}
   \sum_k \frac{|\alpha_n(k) \beta_n(q-k) - 1|^2}{(\frac{1}{M}(q-k)^2 + k^2 - E)^2} < \eps,
\end{equation}
while for the finitely many values of $q \in \kappa \Z^2$ with $|q| \leq C_\eps$, there exists $N_\eps\in \N$ such that \eqref{sup-eps} holds for all $n\geq N_\eps$. In view of  \eqref{Sup_Int_Constant}, this proves that
$$
   |\sprod{w}{(V_n(H_0 - E)^{-1} - B_E) \psi}| < \sqrt{\eps} \cdot \sqrt{N} \cdot \norm{w} \cdot \norm{\psi}
$$
for all $n \geq N_\eps$. Hence $V_n(H_0 - E)^{-1} \to B_E$ as $n \to \infty$. The compactness of $B_E$ is a consequence of the compactness of $(H_0 - E)^{-1}$.
\end{proof}

Next, we verify that Conditions (b) and (c) of \Cref{Convergence_Theorem} are satisfied in the situation of \Cref{Regularization_Theorem}. To this end we write
 $$
     \phi_n(z) := g_n^{-1} - V_n (H_0 - z)^{-1} V_n^* \in \LL(\tilde{\HH}_{N-1})
$$ 
in normal ordered form. Making use of the pull-through formulas \eqref{Pull_Through_Formula}, the definition of $g_n$ given in \eqref{Renormalization_Condition}, and the identity $b_p b_q^* = \delta_{p,q}$ on the vacuum sector, we find that for $z\in\R_{-}\cup(\C\backslash\R)$,
$$
    \phi_n(z) = \phi^0_n(z) + \phi^I_n(z),
$$
where 
\begin{align}
   \phi^0_n(z) &:= \sum_q m_q^* \sum_k \left( \frac{\alpha_n(k)^2 \cdot \beta_n(-k)^2}{(1 + \frac{1}{M}) k^2 - E_B} - \frac{\alpha_n(k)^2 \cdot \beta_n(q-k)^2}{H_f + \frac{1}{M}(q-k)^2 + k^2 - z} \right) m_q \label{Explicit_Formula_Phi} \\
    \phi^I_n(z) &:= \sum_{k,l,q} m_{q+k}^* \: a_l^* \: \frac{\alpha_n(k) \cdot \alpha_n(l) \cdot \beta_n(q)^2}{H_f + \frac{1}{M} q^2 + k^2 + l^2 - z} \: a_k \: m_{q+l} \nonumber
\end{align}
and $H_f := \sum_k k^2 \: a_k^* a_k$.
Let $D \subset \tilde{\HH}$ be the dense subspace of all finite linear combinations of vectors of the form $\ph_q \otimes \ph_{p_1} \wedge\ldots \wedge \ph_{p_{N-1}}$. Recall that $\ph_p \in L^2(\Omega)$ denotes a plane wave with momentum $p \in \kappa \Z^2$ defined by $\ph_p(x) := L^{-1} \cdot \exp(ipx)$.

\begin{lemma} \label{Facts_about_Phi}
Suppose Hypotheses (A) and (B) of \Cref{Regularization_Theorem} are satisfied. Then, with $D$ and $\phi_n(z)$ as defined above:
\begin{enumerate}[label=(\roman*), font=\normalfont]
 \item There is an operator $\phi(z): D \to \tilde{\HH}_{N-1}$ such that $\phi_n(z) w \to \phi(z) w$ as $n \to \infty$ for all $w \in D$ and $z\in \rho(H_0)$. The operator $\phi(z)$ is essentially self-adjoint for $z \in \R \cap \rho(H_0)$, and it is independent of the choice of the sequences $(\alpha_n)$ and $(\beta_n)$. Explicitely, for $z\in \R_{-}\cup(\C\backslash\R)$,
\begin{align}
   \phi(z) =\ &\sum_q m_q^*  \sum_k \left( \frac{1}{(1 + \frac{1}{M}) k^2 - E_B} - \frac{1}{H_f + \frac{1}{M}(q-k)^2 + k^2 - z} \right)m_q\nonumber \\
        &+ \sum_{k,l,q} m_{q+k}^* \: a_l^* \: \frac{1}{H_f + \frac{1}{M} q^2 + k^2 + l^2 - z} \: a_k \: m_{q+l}. \label{explicit-phi}
\end{align}
 \item For every $c > 0$, there is a $\tau_c < 0$ such that $\phi_n(\tau_c) \geq c$ for $n$ sufficiently large.
\end{enumerate}
\end{lemma}

\noindent
\textbf{Remark.} In the following, we use the notation $\phi(z)$ for the closure of $\phi(z)\rst D$.

\begin{proof}
It suffices to prove (i) for $z = \tau\leq -1$. Then, the statement will follow all $z\in\rho(H_0)$ from 
\begin{equation}\label{resolvent-argument}
          \phi_n(z)-\phi_n(\tau) = (\tau - z)\cdot V_n(H_0-z)^{-1}(H_0-\tau)^{-1} V_n^{*}
\end{equation}
and from \Cref{Norm_Convergence_V}. Throughout this proof, we assume that $\gamma_n$ is given by \eqref{gamma_Alternative_1}. For $\gamma_n$ given by \eqref{gamma_Alternative_2}, the proof proceeds along the same lines after the roles of $\alpha_n$ and $\beta_n$ have been interchanged in $\phi_n^0(z)$ by the substitutions  $k \to -k$ and $k \to q-k$ of the summation indices in \eqref{Explicit_Formula_Phi}. 

To prove (i) we set $v := \ph_q \otimes \ph_{p_1} \wedge \ldots\wedge \ph_{p_{N-1}} \in D$ and write $P^2 := p_1^2 + \ldots + p_{N-1}^2$. Then, $\phi^0_n(\tau) v = \mu_{\tau, n}(q, P^2) v$ with
\begin{equation}\label{Definition_Mu_Tau_n}
   \mu_{\tau, n}(q,P^2) = \sum_k \left( \frac{\alpha_n(k)^2 \cdot \beta_n(-k)^2}{(1 + \frac{1}{M}) k^2 - E_B} - \frac{\alpha_n(k)^2 \cdot \beta_n(q-k)^2}{\frac{1}{M}(q-k)^2 + k^2 + P^2 -\tau} \right).
\end{equation}
The eigenvalue $\mu_{\tau,n}(q,P^2)$ is convergent as $n \to \infty$ with limit
\begin{equation} \label{Definition_Mu_Tau}
   \mu_{\tau}(q,P^2) = \sum_k \left( \frac{1}{(1+\frac{1}{M})k^2 - E_B} - \frac{1}{\frac{1}{M}(q-k)^2 + k^2 + P^2  -\tau} \right).
\end{equation}
This can be seen by writing
\begin{align}
   \mu_{\tau,n}(q,P^2) &= \sum_k \alpha_n(k)^2 \beta_n(-k)^2 \left( \frac{1}{(1 + \frac{1}{M})k^2 - E_B} - \frac{1}{\frac{1}{M}(q-k)^2 + k^2 + P^2  -\tau} \right)\nonumber \\
   &\qquad + \sum_k \alpha_n(k)^2 \cdot \frac{\beta_n(-k)^2 - \beta_n(q-k)^2}{\frac{1}{M}(q-k)^2 + k^2 + P^2  -\tau}. \label{phi-ev-2}
\end{align}
The first sum converges to $\mu_\tau(q, P^2)$ by (A) and by dominated convergence, because the term in brackets is $O(|k|^{-3})$ for $|k| \to \infty$. The second sum can be estimated from above in absolute value by a constant times $\gamma_n(q)$ because for  $\tau\leq -1$
\begin{equation}\label{Denominator_Estimate}
   \frac{1}{M}(q-k)^2 + k^2 + P^2  -\tau \geq c (k^2 + q^2 + 1),
\end{equation}
with some $c>0$.  Inequality \eqref{Denominator_Estimate} follows from the fact that $M^{-1}(q-k)^2 + k^2$ is homogeneous of degree two and strictly positive on the compact set $\{(k,q)\mid k^2+ q^2=1\}$. Thus, by (B), the second term of \eqref{phi-ev-2}, like $\gamma_n(q)$, vanishes in the limit $n\to\infty$.

We conclude that $\phi^0_n(\tau) w$ converges for all $w \in D$, and we denote the limit by $\phi^0(\tau) w$. The operator $\phi^0(\tau)$ is essentially self-adjoint on $D$, since the vectors of the form $\ph_q \otimes \ph_{p_1} \wedge\ldots \wedge \ph_{p_{N-1}}$ form a basis of eigenvectors associated with real eigenvalues.

By definition of $\phi^I_n(\tau)$ and $v$,
\begin{multline}
 \phi^I_n(\tau) v = \sum_l \sum_{j=1}^{N-1} (-1)^{j-1} \Bigg(\frac{\alpha_n(p_j) \alpha_n(l) \beta_n(q - l)^2}{\frac{1}{M}(q-l)^2 + l^2 + P^2  -\tau} \\ 
 \times\ph_{q+p_j-l} \otimes \ph_l \wedge \ph_{p_1} \wedge \ldots \wedge \widehat{\ph_{p_j}} \wedge \ldots \wedge \ph_{p_{N-1}}\Bigg), \label{Phi^I_n_v}
\end{multline}
where the notation $\widehat{\ph_{p_j}}$ means that the vector $\ph_{p_j}$ is omitted.
This is an expansion in an orthonormal system, where, by (A), the coefficients have a limit as $n\to\infty$, and by the uniform boundedness of the numerators, they have a square summable majorant. If follows that $\phi^I_n(\tau)v$, in the limit $n\to\infty$, converges to
\begin{equation}
 \sum_l \sum_{j=1}^{N-1} \frac{(-1)^{j-1}}{\frac{1}{M}(q-l)^2 + l^2 + P^2  -\tau}\\
 \ph_{q+p_j-l} \otimes \ph_l \wedge \ph_{p_1} \wedge \ldots \wedge \widehat{\ph_{p_j}} \wedge \ldots \wedge \ph_{p_{N-1}}.\label{PhiI_Product_State}
\end{equation}

Next, we show that the operators $\phi_n^I(\tau)$ are uniformly bounded in $n$ and $\tau\leq -1$. Let $w \in \tilde{\HH}_{N-1}$. By definition of  $\phi_n^I(\tau)$, $\sprod{w}{\phi_n^I(\tau) w}$ is bounded in absolute value by a constant times
\begin{align}   
   &\sum_{k,l,q} \norm{a_l m_{q+k} w} \cdot \norm{(H_f + \frac{1}{M} q^2 + k^2 + l^2 -\tau)^{-1} a_k m_{q+l} w} \label{Start_Estimate_Phi_I} \\
   &\leq \sum_{k,l} \frac{1}{k^2 + l^2 +1}\sum_q\norm{a_l m_{q+k} w} \cdot \norm{a_k m_{q+l} w} \nonumber \\
   &\leq \sum_{k,l}  \frac{1}{k^2 + l^2 +1} \left( \sum_q \norm{a_l m_{q+k} w}^2 \right)^{1/2} \left( \sum_q \norm{a_k m_{q+l} w}^2 \right)^{1/2} 
    =\sum_{k,l} \frac{\norm{a_l w} \cdot \norm{a_k w}}{k^2 + l^2 +1} \nonumber \\
   &\leq \left( \sum_{k,l} \norm{a_l w}^2 \frac{(l^2 +1)^{1/2}}{(k^2 +1)^{1/2} (k^2+l^2+1)} \right)^{1/2} \left( \sum_{k,l} \norm{a_k w}^2 \frac{(k^2+1)^{1/2}}{(l^2 +1)^{1/2} (k^2 + l^2 +1)} \right)^{1/2} \nonumber \\
   &\leq (N-1) \cdot \norm{w}^2 \cdot \sup_{l \in \kappa \Z^2} \sum_{k} \frac{(l^2 +1)^{1/2}}{(k^2 +1)^{1/2} (k^2+l^2+1)}. \nonumber
\end{align}
Upon a comparison of this sum with the integral
$$
     \int_{\R^2}\frac{1}{|k|(k^2+l^2+1)}dk =\frac{\pi^2}{\sqrt{l^2+1}},
$$
see \Cref{Riemann}, we conclude that,  
\begin{equation}\label{Phi_I_Bound}
   \sup_{n\in\N}\norm{\phi_n^I(\tau)} \leq \const \cdot (N-1), \qquad \tau\leq -1.
\end{equation}
From the convergence of \eqref{Phi^I_n_v}, which  extends to all $v \in D$, and from the uniform bound \eqref{Phi_I_Bound}, it follows that $\phi^I_n(-1)$ is strongly convergent to an operator $\phi^I(-1) \in \LL(\tilde{\HH})$. Since $\phi^0(-1)$ is essentially self-adjoint on $D$ and $\phi^I(-1)$ is bounded and symmetric, we see that $\phi(-1): D \to \tilde{\HH}$ with $\phi(-1) = \phi^0(-1) + \phi^I(-1)$, which is the strong limit of $\phi_n(-1)$ on $D$ as $n\to\infty$, is essentially self-adjoint. In view of \eqref{resolvent-argument} the proof of (i) is complete.

To prove (ii) we show that there is a $c > 0$, independent of $\tau\leq -1$, such that
\begin{equation} \label{First_Step_For_Positivity}
   \mu_{\tau,n}(q, P^2) \geq \mu_{\tau, n}(0,0) - c 
\end{equation}
for all $q \in \kappa \Z^2$, $P^2 \geq 0$ and $n \in \mathbb{N}$. Since $\mu_{\tau, n}(q, P^2) \geq \mu_{\tau, n}(q,0)$, it suffices to verify that $\mu_{\tau, n}(q,0) - \mu_{\tau, n}(0,0)$ is bounded from below uniformly in $n$ and $q$. By the convergence of $\mu_{\tau, n}(q,0)$ in the limit $n \to \infty$, we may omit finitely many values of $q\in \kappa\Z^2$ and assume that $|q| > 4\sqrt{2}\kappa$. We write
\begin{eqnarray}
 \lefteqn{\mu_{\tau,n}(q,0) - \mu_{\tau,n}(0,0)}\nonumber\\
 &&= \sum_k \alpha_n(k)^2 \beta_n(-k)^2 \left( \frac{1}{(1+\frac{1}{M})k^2 - \tau} - \frac{1}{\frac{1}{M}(q-k)^2 + k^2 - \tau} \right) \nonumber\\
 &&\qquad + \sum_k \alpha_n(k)^2 \cdot \frac{\beta_n(-k)^2 - \beta_n(q-k)^2}{\frac{1}{M}(q-k)^2 + k^2 - \tau}. \label{Decomposition_for_Positivity}
\end{eqnarray}
The difference of the two denominators in the first sum equals $M^{-1}(q^2-2kq)$, which is non-negative for $|k| \leq |q|/2$. Therefore, the first sum in \eqref{Decomposition_for_Positivity} is bounded from below by
\begin{align*}
 &\frac{1}{M} \sum_{|k| > |q|/2} \alpha_n(k)^2 \beta_n(-k)^2 \frac{q^2 - 2 k \cdot q}{((1+\frac{1}{M})k^2 - \tau)(\frac{1}{M}(q-k)^2 + k^2 - \tau)} \\
 &\geq -\frac{2}{M} \sum_{|k| > |q|/2} \alpha_n(k)^2 \beta_n(-k)^2 \frac{|k|\cdot |q|}{((1+\frac{1}{M})k^2 - \tau)(\frac{1}{M}(q-k)^2 + k^2 - \tau)} \\
 &\geq - \frac{2}{M} \sum_{\substack{k \in \kappa \Z^2 \\ |k| > |q|/2}} \frac{|q|}{(k^2 + 1)^{3/2}},
\end{align*}
where $-\tau \geq 1$ was used and some positive terms were dropped in the denominator. The remaining sum can be estimated in terms of integrals over the $|k|$-range $|q|/2-\sqrt{2}\kappa<|k|<\infty$, which  is contained in $|q|/4<|k|<\infty$ because $\sqrt{2}\kappa<|q|/4$ by assumption. See the proof of \Cref{Riemann}. From these integrals we conclude that the first sum in \eqref{Decomposition_for_Positivity} is bounded from below uniformly in $q$, $n$ and $\tau \leq -1$. The same is true for the second sum in \eqref{Decomposition_for_Positivity}, because, by \eqref{Denominator_Estimate}, it can be estimated in terms of $\gamma_n(q)$, which is uniformly bounded by Hypothesis (B). Thus, \eqref{First_Step_For_Positivity} is proved.

As already mentioned, vectors of the form $\ph_q \otimes \ph_{p_1} \wedge\ldots \wedge \ph_{p_{N-1}}$ with $q, p_1,\ldots, p_{N-1} \in \kappa \Z^2$ form a total set of eigenvectors of $\phi_n^0(\tau)$ with eigenvalues $\mu_{\tau, n}(q, P^2)$. Hence, by  \eqref{First_Step_For_Positivity},
$$
   \phi_n^0(\tau) \geq \inf_{q \in \kappa \Z^2, P^2 > 0} \mu_{\tau, n}(q, P^2) \geq \mu_{\tau,n}(0,0) - c.
$$
Statement (ii) now follows from the uniform bound \eqref{Phi_I_Bound} on $\phi_n^I(\tau)$, from the convergence $\mu_{\tau,n}(0,0) \to \mu_\tau(0,0)$ as $n \to \infty$, and from $\mu_\tau(0,0) \to \infty$ as $\tau \to -\infty$.
\end{proof}

\begin{proof}[Proof of \Cref{Regularization_Theorem}]
The existence statement of \Cref{Regularization_Theorem} is a consequence of \Cref{Convergence_Theorem}, \Cref{Norm_Convergence_V} and \Cref{Facts_about_Phi}. 
 Equation~\eqref{Resolvent_Limit_Operator} of \Cref{Regularization_Theorem} expresses the resolvent $(H-\mu)^{-1}$ in terms of the operators $B_\mu$ and $\phi(\mu)$, which are independent of the sequences $(\alpha_n)$ and $(\beta_n)$. This completes the proof of \Cref{Regularization_Theorem}.
\end{proof}

We conclude this section by emphasizing that Assumptions (H1) and (H2) hold for the Hamiltonian $H$ constructed in \Cref{Regularization_Theorem}. The validity of (H1) is obvious and (H2) follows from the following lemma. Consequently, the results of Section 6 apply to this case including the Birman-Schwinger principle of \Cref{Minimum_Expectation_Value}. 

\begin{lemma} \label{Notin_Lemma}
 Let $w \in \tilde{\HH}_{N-1}$ with $w \neq 0$. Then $B_z^* w \notin D(H_0^{1/2})$ for all $z \in \rho(H_0)$. In particular, $\ker B_z^* \cap \tilde{\HH}_{N-1} = \{ 0 \}$.
\end{lemma}

\begin{proof}
Assume that $B_{-1}^* w \in D(H_0^{1/2})$ for some $w \in \tilde{\HH}$. By \Cref{Norm_Convergence_V},  $B_{-1}^{*}=\lim_{n\to\infty}(1 + H_0)^{-1}V_{n}^*$, where $V_n$ is defined by \eqref{def:V_n} with arbitrary cutoffs $\alpha_n, \beta_n$ satisfying the hypotheses of \Cref{Regularization_Theorem}. Therefore,
\begin{align} 
  \norm{(1 + H_0)^{1/2} B_{-1}^* w}^2 
  &= \lim_{\eps \downarrow 0} \: \norm{(1 + \eps H_0)^{-1/2} (1 + H_0)^{1/2} B_{-1}^* w}^2 \nonumber\\
 &= \lim_{\eps \downarrow 0} \lim_{n \to \infty} \norm{(1 + \eps H_0)^{-1/2} (1 + H_0)^{1/2} (1 + H_0)^{-1} V_{n}^* w}^2 \nonumber\\
  &= \lim_{\eps \downarrow 0} \: \lim_{n \to \infty} \: \sprod{w}{V_n (1 + H_0)^{-1} (1 + \eps H_0)^{-1} V_{n}^* w}. \label{Notin_Proof_1}
\end{align}
By \eqref{resolvent-argument}, for all $\eps\in (0,1)$,
\begin{align*}
   V_n (1 + H_0)^{-1} (1 + \eps H_0)^{-1} V_{n}^* &= (\phi_n(-1/\eps) - \phi_n(-1))\frac{1}{1-\eps}\\
&\geq (\phi_n^{0}(-1/\eps) - \phi_n^{0}(-1))\frac{1}{1-\eps} - \frac{2C}{1-\eps},
\end{align*}
where $C=\sup_{n\in\N,\tau\leq -1} \|\phi_n^{I}(\tau)\|<\infty$ by \eqref{Phi_I_Bound}. Recall from the proof of \Cref{Facts_about_Phi} that every vector $v := \ph_q \otimes \ph_{p_1} \wedge \ldots\wedge \ph_{p_{N-1}}$ is an eigenvector of $\phi_n^{0}(\tau)$ with eigenvalue $\mu_{\tau,n}(q,P^2)$ given by \eqref{Definition_Mu_Tau_n}. The difference $\mu_{-1/\eps,n}(q,P^2)-\mu_{-1,n}(q,P^2)$ is positive for $\eps<1$ and all $n\in\N$. For $n\to\infty$ this difference has the limit $\mu_{-1/\eps}(q,P^2)-\mu_{-1}(q,P^2)$, which, by \eqref{Definition_Mu_Tau}, diverges for $\eps\to 0$. It follows that 
$$
   \norm{(1 + H_0)^{1/2} R_{-1}^* w}^2 \geq  \lim_{\eps \downarrow 0} \frac{1}{1-\eps} \Big( |\sprod{w}{v}|^2\big(\mu_{-1/\eps}(q,P^2)-\mu_{-1}(q,P^2)\big) - 2C \Big)=\infty,
$$
unless $\sprod{w}{v}=0$. Since the eigenvectors $v$ form an ONB of $\tilde{\HH}_{N-1}$, it follows $w=0$. This proves the lemma for $z=-1$. To prove it for general  $z \in \rho(H_0)$ it suffices to note that, by a resolvent identity,
$$
    B_z^* w = B_{-1}^* w + (\overline{z} + 1) (H_0 - \overline{z})^{-1} B_{-1}^* w,
$$
where the second summand belongs to $D(H_0)\subset D(H_0^{1/2})$.
\end{proof}

The following lemma was often used in the present section.

\begin{lemma}\label{Riemann}
Let $f:[0,\infty)\to[0,\infty)$ be monotonically decreasing. Then 
$$
     \sum_{k\in \kappa\Z^2}f(k^2) \leq f(0) + \frac{4}{\kappa}\int_0^{\infty}f(t^2)\,dt + \frac{2\pi}{\kappa^2}\int_0^{\infty} f(t^2)t\, dt.
$$
\end{lemma}
 
\begin{proof}
By the symmetry of the function $k\mapsto f(k^2)$,
$$
  \sum_{k\in \kappa\Z^2}f(k^2) = f(0) + 4 \sum_{k\in \kappa(\Z_{+}\times\{0\})}f(k^2) + 4 \sum_{k\in \kappa\Z_{+}^2}f(k^2),
$$
which, by the monotonicity of $f$ is bounded from above by the integrals given in the statement.
\end{proof}

\section{The polaron and molecule states}
\label{Sect:Pol_and_Mol_Equation}

Two trial states that are investigated intensively in the physics literature are the polaron ansatz and the molecule ansatz at first order in a particle-hole expansion. The polaron and the molecule ansatz are expected to approximate the ground state of the Fermi polaron well in the case of weak and strong coupling between the impurity and the Fermi gas, respectively \cite{CombescotGiraud}.

Fix $\mu > 0$. Let $\ket{\textrm{FS}_\mu}$ denote the Fermi sea with Fermi energy $\mu$, which is given by
\begin{equation} \label{Fermi_Sea}
   \ket{\textrm{FS}_\mu} = \prod_{k^2 \leq \mu} a^*_k \ket{\text{vac}},
\end{equation}
where $\ket{\text{vac}}$ is the vacuum state. Let
$$
   N_\mu := |\{ k \in \kappa \Z^2 \, : \, k^2 \leq \mu \}|.
$$
denote the number of fermions in $\ket{\textrm{FS}_\mu}$ and let $E_\mu$ be the kinetic energy of the Fermi sea, i.e.
$$
   H_f \ket{\textrm{FS}_\mu} = E_\mu \ket{\textrm{FS}_\mu}
$$
with $H_f := \sum_k k^2 \: a_k^* a_k$. The total momentum of the Fermi sea vanishes. In fact, $P_f \ket{\textrm{FS}_\mu} = 0$ with $P_f := \sum_k k \: a_k^* a_k$.

The trial state which is often referred to as polaron ansatz was first proposed by Chevy \cite{Chevy}. It is represented by
\begin{equation} \label{Polaron_Physics}
   \ket{\textrm{P}} = \alpha_0 b_0^* \ket{\textrm{FS}_\mu} + \sum_{\substack{K^2 > \mu \\ q^2 \leq \mu}} \alpha_{K, q} b^*_{q-K} a^*_K a_q \ket{\textrm{FS}_\mu}
\end{equation}
with coefficients $\alpha_0, \alpha_{K, q} \in \C$.
The polaron trial state \eqref{Polaron_Physics} is a state of total momentum zero with $N_\mu$ fermions. It consists of the ground state $b_0^* \ket{\textrm{FS}_\mu}$ of the kinetic energy on $\HH_{N_\mu}$, and the first term in the so-called particle-hole expansion. The action of the operator $a_q$ with $q^2 \leq \mu$ can be interpreted as the creation of a ``hole'' in the Fermi sea and $a_K^*$ with $K^2 > \mu$ creates a ``particle'' with momentum outside the Fermi sphere $k^2=\mu$.

The molecule ansatz was proposed independently by Chevy and Mora \cite{ChevyMora} and by Punk, Dumitrescu and Zwerger \cite{PDZ}. It reads
\begin{equation} \label{Molecule_Physics}
   \ket{\textrm{M}} = \sum_{K^2 > \mu} \beta_K b^*_{-K} a^*_K \ket{\text{FS}_\mu} + \sum_{\substack{K^2,L^2 > \mu \\ q^2 \leq \mu}} \beta_{K,L,q} b^*_{q-K-L} a^*_K a^*_L a_q \ket{\text{FS}_\mu}
\end{equation}
with coefficients $\beta_K, \beta_{K,L,q} \in \C$. Compared to the polaron ansatz, it is also a state of total momentum zero, but with $N_\mu + 1$ fermions. 

Since there is no explicit expression for $H$, it is not possible to compute $\sprod{P}{HP}$ and  $\sprod{M}{HM}$ directly. In the physics literature this problem is avoided by computing the regularized expressions  $\sprod{P}{H_nP}$ and  $\sprod{M}{H_nM}$ instead. The conditions of stationarity subject to the normalization conditions for the states $|P\rangle$ and $|M\rangle$  lead to implicit equations for the Langrange multiplier $E$, from which the regularization and (some of) the variational parameters can be eliminated. The solutions for the Lagrange multipliers in these equations are called \emph{polaron energy}, $E_P$, and \emph{molecule energy}, $E_M$, respectively, and these numbers are considered upper bounds for the ground state energy of the system.  This procedure lacks any justification in the absence of a definition for the (non-regularized) Hamiltonian.

In this section, we apply the variational principle established in \Cref{Sect:Variational_Principle} (\Cref{Convergence_Theorem}) to prove that the polaron energy $E_P$ given, e.g., by (4) in \cite{Parish} and the molecule energy $E_M$ given by (6) and (7) in \cite{Parish} yield upper bounds to the ground state energy of $H$ on $\HH_{N_\mu}$ and $\HH_{N_\mu + 1}$, respectively. Since the polaron and the molecule ansatz have non-equal numbers of fermions while the Fermi energies (or chemical potentials) agree, the proper observable to minimize is $H-\mu N$, rather that $H$. This means that the polaron ansatz should be considered a better approximation to the ground state of the Fermi polaron than the molecule ansatz, if
\begin{equation} \label{Polaron_Better}
   E_P < E_M - \mu,
\end{equation}
and the molecule ansatz should be considered a better approximation to the ground state of the Fermi polaron than the polaron ansatz, if
\begin{equation} \label{Molecule_Better}
   E_P > E_M - \mu.
\end{equation}

\subsection*{Chevy's polaron equation}

In this and the following section we use the variational principle for the operator $\phi(E)$ to show that the polaron and molecule energies are indeed upper bounds to the ground state energy of $H_{N_\mu}$ and $H_{N_\mu+1}$, respectively. Since we work with $\phi(E)$ rather than $H$, we need analogs of \eqref{Polaron_Physics} and \eqref{Molecule_Physics} in $\HH_{N_\mu-1}$ and $\HH_{N_\mu}$, respectively. In the case of the polaron, our choice is 
\begin{equation} \label{Polaron_Linden}
   | \widetilde{\textrm{P}} \rangle = \sum_{\substack{q^2 \leq \mu}} \tilde{\alpha}_q m_q^* a_q \ket{\textrm{FS}_\mu},
\end{equation}
which will be justified by \Cref{Polaron_Equation}, below. Note that $ | \widetilde{\textrm{P}} \rangle$ depends on $N_\mu$ free parameters, while \eqref{Polaron_Physics} has infinitely many variational parameters.

According to \Cref{Upper_Bound}, any solution $E$  to
\begin{equation}\label{new-var-polaron}
   \min_{\lVert \widetilde{P} \rVert = 1} \sprod{\widetilde{\textrm{P}}}{\phi(E) \widetilde{\textrm{P}}} = 0
\end{equation}
is an upper bound to the ground state energy of $H$ on $\HH_{N_\mu}$. In order to compute expectation values of $\phi(E)$ with respect to excitations of the Fermi sea such as \eqref{Polaron_Linden}, it is helpful to first invert the normal ordering of $a_l^*$ and $a_k$ in \eqref{explicit-phi} if $k^2 \leq \mu$ or $l^2 \leq \mu$. By the pull-through formula this leads to
\begin{align}
   \phi(z) &= \sum_q m_q^* G_\mu(H_f-z,q) m_q + \sum_{\substack{q \\ l^2 \leq \mu<k^2}} \left(m_{q+k}^* \frac{1}{H_f + \frac{1}{M}q^2 + k^2 - z} a_l^* a_k m_{q+l} + \textrm{h.c.} \right) \nonumber\\
   &+ \sum_{\substack{q \\ k^2,l^2 > \mu}} \!\!\! m_{q+k}^* a_l^* \frac{1}{H_f + \frac{1}{M} q^2 + k^2 + l^2 - z} a_k m_{q-l} - \sum_{\substack{q \\ k^2,l^2 \leq \mu}} \!\!\! m_{q+k}^* a_k \frac{1}{H_f + \frac{1}{M} q^2 - z} a_l^* m_{q+l} \label{Fermi_Sea_Phi}
\end{align}
for $z\in \R_{-}\cup(\C\backslash\R)$, with
\begin{equation} \label{Function_G}
   G_\mu(\lambda, q) := \sum_k \left( \frac{1}{(1+\frac{1}{M}) k^2 - E_B} - \frac{\chi(k^2 > \mu)}{\frac{1}{M}(q-k)^2 + k^2 + \lambda} \right)
\end{equation}
for $\lambda \in \R$ and $q \in \R^2$. Only the first and the fourth terms of \eqref{Fermi_Sea_Phi} give contributions to the matrix elements in \eqref{new-var-polaron}. Explicitly,
$$
   \sprod{\widetilde{\textrm{P}}}{\phi(E) \widetilde{\textrm{P}}}  = \sum_{q^2\leq \mu}  |\tilde{\alpha}_q|^2 G_\mu(E_\mu-E-q^2, q) - \frac{1}{E_\mu-E}\sum_{q^2,l^2\leq \mu}  \tilde{\alpha}_q^{*} \tilde{\alpha}_l,
$$ 
which is valid for $E<E_\mu$ by analytic continuation. The expression on the right hand side is a quadratic form in $(\tilde{\alpha}_q)_{q^2\leq \mu}$ depending on $\lambda=E_\mu-E$. We are going to write it as a quadratic form in the space  $\hh_\mu=\linhull\{\ph_q:q^2\leq \mu\}$. To this end let $ \tilde{\alpha}=\sum_{q^2 \leq \mu}   \tilde{\alpha}_q \ph_q $, let $\xi := \sum_{q^2 \leq \mu} \ph_q$, and let  $T(\lambda)$ be the linear operator in $\hh_{\mu}$ that, in the basis $\{\ph_q\mid q^2\leq \mu\}$, is diagonal with eigenvalues $G_\mu(\lambda - q^2, q)$. Then, for $E<E_\mu$, $\sprod{\widetilde{\textrm{P}}}{\phi(E) \widetilde{\textrm{P}}} = \sprod{ \tilde{\alpha}}{P(E_\mu-E) \tilde{\alpha}} $, where
$$
P(\lambda) :=  T(\lambda) - \frac{1}{\lambda}\ket{\xi} \bra{\xi}.
$$

The operator $P(\lambda)$ is of the form of the operators considered in the example of Section~\ref{sec:BS-operator}, and its Birman-Schwinger operator, which is a number, for $z=0$ reads
$$
       \lambda -  \sprod{\xi}{T(\lambda)^{-1}\xi}  = \lambda -  \sum_{q^2 \leq \mu} G_\mu(\lambda - q^2, q)^{-1}.
$$
This explains much of the following proposition.

\begin{prop}[Polaron Equation]\label{Polaron_Equation}
For $\lambda>0$ and $\ell \in \N$,
$$
   \mu_\ell(P(\lambda)) = 0 \quad \Rightarrow \quad \mu_\ell(H) \leq E_\mu-\lambda.
$$
In particular, any solution $\lambda$ to the polaron equation
\begin{equation} \label{Chevy_Equation}
   \lambda = \sum_{q^2 \leq \mu} G_\mu(\lambda- q^2, q)^{-1}
\end{equation}
defines an upper bound $E_\mu-\lambda$ to the ground state energy of $H$ on $\HH_{N_\mu}$. Equation \eqref{Chevy_Equation} has at least one solution $\lambda>0$ and the largest solution is characterized by  $\mu_1(P(\lambda)) = 0$. A (non-normalized) trial state of the form \eqref{Polaron_Linden} associated with the largest solution $\lambda$ has the coefficients
\begin{align}\label{Optimal_Trial_State}
   \tilde{\alpha}_q = G_\mu(\lambda- q^2, q)^{-1}.
\end{align}
\end{prop}

\noindent
\emph{Remark.} Equation \eqref{Chevy_Equation} with $\lambda=E_\mu-E_P$ agrees with the implicit equation (4) in \cite{Parish} for the polaron energy $E_P$, see also \cite{Chevy}. \Cref{Polaron_Equation} explains the meaning of this equation and it justifies our choice \eqref{Polaron_Linden} for the polaron trial states.

\begin{proof}
Let $I$ denote the linear isometry $I:\hh_\mu\to \HH_{N_\mu-1}$ defined by $I\ph_q = m_q^* a_q \ket{\textrm{FS}_\mu}$. Let $E=E_\mu-\lambda$.
From the equation $\sprod{\beta}{P(\lambda)\beta}=\sprod{I \beta}{\phi(E) I \beta}$ it is clear that $\mu_\ell(P(\lambda))\geq \mu_\ell(\phi(E))$ for $\ell=1\ldots N_\mu$. Hence, by  \Cref{Upper_Bound}, $\mu_\ell(P(\lambda))\leq 0$ implies that $\mu_\ell(H)\leq E=E_\mu-\lambda$ provided that $\lambda>0$.

The connection between $ \mu_\ell(P(\lambda)) = 0$ and  the equation \eqref{Chevy_Equation}  is explained by the example  of Section \ref{sec:BS-operator}: since $\lambda - \sum_{q^2 \leq \mu} G_\mu(\lambda - q^2, q)^{-1}$ is the Birman-Schwinger operator associated with $P(\lambda)$ and $z=0$, \eqref{Chevy_Equation} implies that $\mu_\ell(P(\lambda))=0$ for some $\ell$. Conversely, if $\mu_\ell(P(\lambda))=0$ and $G_\mu(\lambda - q^2, q)\neq 0$ for all $q$, $q^2\leq \mu$, then \eqref{Chevy_Equation}  holds.  The condition $G_\mu(\lambda - q^2, q)\neq 0$ is satisfied at least for $\ell=1$, because $\mu_1(P(\lambda))<\mu_1(T(\lambda))$. Indeed, if $\mu_1(T(\lambda))=G_\mu(\lambda-q_0^2,q_0)$ by choice of $q_0$, then 
$$
     \mu_1(P(\lambda)) \leq \sprod{\ph_{q_0}}{P(\lambda)\ph_{q_0}} = \sprod{\ph_{q_0}}{T(\lambda)\ph_{q_0}} - \frac{1}{\lambda}
     = \mu_1(T(\lambda)) - \frac{1}{\lambda}.
$$
Next we show that there exists $\lambda>0$ such that $\mu_1(P(\lambda)) = 0$. It follows from \eqref{Function_G}, or from \eqref{PhiE_Phiz_kurz} and \eqref{RE_Rz}, that $\lambda \mapsto P(\lambda)$ is a continuous matrix valued function on $\R_{+}$. Hence, the eigenvalues $\mu_\ell(P(\lambda))$ are continuous functions of $\lambda$. From \Cref{Facts_about_Phi} we see that $\mu_1(P(\lambda)) \to \infty$ as $\lambda \to \infty$. On the other hand, $\mu_1(P(\lambda)) \to -\infty$ as $\lambda \searrow 0$, because  $\mu_1(P(\lambda))\leq \sprod{\ph_0}{P(\lambda)\ph_0} = G_\mu(\lambda,0) - \lambda^{-1}$ where $G_\mu(\lambda,0)$ is continuous in $\lambda=0$. Hence, there exists $\lambda>0$ such that $\mu_1(P(\lambda)) = 0$. By the example of Section~\ref{sec:BS-operator}, an eigenvector of $P(\lambda)$ belonging to the eigenvalue $\mu_1(P(\lambda))=0$ is given by $T(\lambda)^{-1}\xi$. Its coefficients are $\tilde{\alpha}_q=\sprod{\ph_q}{T(\lambda)^{-1}\xi} = G_\mu(\lambda - q^2, q)^{-1}$.

Finally, let $\lambda$ be the largest solution to  \eqref{Chevy_Equation} and suppose that $\mu_\ell(P(\lambda))=0$ for some $\ell>1$, while   $\mu_1(P(\lambda))<0$. Then, by the arguments above, $\mu_1(P(\lambda'))=0$ for some $\lambda'>\lambda$. It follows that $\lambda'$ is a solution to \eqref{Chevy_Equation},  which contradicts the assumption on $\lambda$. 
\end{proof}

\noindent\emph{Remark.} In the proof above we have shown that $\mu_1(P(\lambda))<\mu_1(T(\lambda))$. For general $\ell\geq 2$ we have
\begin{equation}\label{eval-rank}
      \mu_{\ell-1}(T(\lambda))\leq\mu_\ell(P(\lambda))\leq \mu_{\ell}(T(\lambda))
\end{equation}
due to the fact that $P(\lambda)$ is a rank-one perturbation of $T(\lambda)$. The symmetry of $q\mapsto G_\mu(\lambda - q^2, q)$ implies that all these eigenvalues of $T(\lambda)$ are degenerate, with the possible exception of $G_\mu(\lambda,0)$. Hence if $\mu_\ell(P(\lambda))=0$ for some $\ell\geq 2$, then, in view of \eqref{eval-rank}, it is very likely that $G_\mu(\lambda - q^2, q)=0$ for some $q$ and the polaron equation \eqref{Chevy_Equation} is not defined. In particular the statement of \Cref{Polaron_Equation} about $\mu_1(P(\lambda))$ will not generalize to $\ell\geq 2$. In \cite{Diss-Linden} this is verified explicitly for $\ell=2$ and $M=1$.

\subsection*{The molecule ansatz}
The molecule ansatz in the physics literature is given by \eqref{Molecule_Physics}. We show that a solution $E_M$ to the energy equations associated with the molecule ansatz (see (6) and (7) in \cite{Parish}) yields an upper bound to the ground state energy of $H$ on $\HH_{N_\mu + 1}$. The general argument is the same as in the case of the polaron ansatz and in particular it is based on  \Cref{Upper_Bound}.

The crucial difference between the polaron and the molecule ansatz is the number of fermions. For the polaron ansatz we found the representation \eqref{Polaron_Linden} in the Hilbert space $\HH_{N_\mu-1}$, which served as a trial state for $\phi(E)$. In analogy to the original trial states \eqref{Polaron_Physics} and \eqref{Molecule_Physics}, we expect the molecule ansatz for the operator $\phi(E)$ to have one fermion more than the polaron ansatz \eqref{Polaron_Linden} and to be represented by a state in $\HH_{N_\mu}$.

In fact, it turns out that an appropriate ansatz for the molecule trial state for the operator $\phi(E)$ is given by
\begin{equation} \label{Form_Molecule}
   |\widetilde{\textrm{M}}\rangle = m_0^*\ket{\text{FS}_\mu} + \sum_{q^2 \leq \mu} \sum_{K^2 > \mu} \gamma_{Kq} m_{q-K}^* a_K^* a_q \ket{\text{FS}_\mu},
\end{equation}
with $\gamma_{Kq} \in \C$. The state $| \widetilde{\textrm{M}} \rangle$ will not be normalized in general, because the coefficient of the first term is fixed to $1$. Computing the expectation value of $\phi(E)$  in the state $| \widetilde{\textrm{M}} \rangle$ using \eqref{Fermi_Sea_Phi}, which is valid for $z\in \R_{-}\cup(\C\backslash\R)$, and analytic continuation yields
\begin{align}
 \sprod{\widetilde{\textrm{M}}}{\phi(E) \widetilde{\textrm{M}}}
 = G_\mu(E_\mu-E, 0) &+ \sum_{K^2 > \mu,\, q^2 \leq \mu} (\gamma_{Kq} + \overline{\gamma}_{Kq}) \frac{1}{(1 + \frac{1}{M})K^2 + E_\mu - E}\nonumber\\ 
 &+ \sum_{K^2 > \mu,\, q^2 \leq \mu} |\gamma_{Kq}|^2 G_\mu(K^2 \!-\! q^2 \!+\! E_\mu-E, q \!-\! K) \nonumber\\*
 &+ \sum_{\substack{K^2, L^2 > \mu, \\ q^2 \leq \mu}}\! \frac{\overline{\gamma}_{Lq} \gamma_{Kq}}{\frac{1}{M}(q-K-L)^2 + K^2 + L^2 - q^2 + E_\mu - E}\nonumber\\
 &-\sum_{\substack{K^2 > \mu, \\ p^2, q^2 \leq \mu}} \frac{\overline{\gamma}_{Kq} \gamma_{Kp}}{(1+\frac{1}{M}) K^2 + E_\mu - E} \label{Molecule_Expectation_Value}
\end{align}
for $E<E_\mu+\mu$. We look for critical points of \eqref{Molecule_Expectation_Value} as function of the parameters $\gamma_{Kq}$ and obtain the condition
\begin{align}
 0 &= \frac{1}{(1+\frac{1}{M})K^2 + E_\mu - E} + \gamma_{Kq} \cdot G_\mu(K^2 \!-\! q^2 \!+\! E_\mu-E, q \!-\! K) \nonumber\\
   &\quad +\!\! \sum_{L^2 > \mu} \gamma_{Lq} \cdot \frac{1}{\frac{1}{M}(q-K-L)^2 + K^2 + L^2 - q^2 + E_\mu - E}\nonumber\\
   &\quad - \sum_{p^2 \leq \mu} \gamma_{Kp} \cdot \frac{1}{(1+\frac{1}{M})K^2 + E_\mu - E} \label{Molecule_Equation_2}
\end{align}
for all $K,q \in \kappa \Z^2$ with $q^2 \leq \mu$ and $K^2 > \mu$. We multiply this equation by $\overline{\gamma}_{Kq}$, sum both sides of it over $K$ and $q$ and combine it with the equation $\sprod{\widetilde{\textrm{M}}}{\phi(E) \widetilde{\textrm{M}}} = 0$ to get
\begin{equation}
 G_\mu(E_\mu-E,0) +\! \sum_{\substack{K^2 > \mu, \\ q^2 \leq \mu}}\! \gamma_{Kq} \frac{1}{(1 + \frac{1}{M})K^2 + E_\mu - E} = 0 \label{Molecule_Equation_1}.
\end{equation}
According to  \Cref{Upper_Bound}, a solution $(E, \{\gamma_{Kq}\})$ of \eqref{Molecule_Equation_2} and \eqref{Molecule_Equation_1} with $E < E_\mu$ is an upper bound $E$ for the ground state energy of $H$. Note that the equations \eqref{Molecule_Equation_2} and \eqref{Molecule_Equation_1} coincide with the equations for the molecule ground state energy (c.f. (6) and (7) in \cite{Parish}), if $E$ is replaced by $E + \mu$. The necessity of this modification was discussed in \eqref{Polaron_Better} and \eqref{Molecule_Better}.


\section{The Fermi-polaron in $\R^2$}
\label{sec:all-space}

In this section we  construct the Hamiltonian for the Fermi polaron in $\R^2$. This is another application of the general framework presented in Section~\ref{Sect:Strong_Resolvent_Limit}. To avoid the introduction of angel-particles, we pass to center-of-mass and relative coordinates, where the regularized Hamiltonian for fixed center-of-mass momentum $P$ has the structure that is required by \Cref{Convergence_Theorem}.

Let $\eta \in L^2(\R^2)$ be compactly supported, $\eta\geq 0$, $\eta(x) = \eta(-x)$ and $\int\eta(x)dx = 2\pi$. Let $\eta_n(x) := n^2 \eta(nx)$ for $n\in \N$. Then, $\int\eta_n(x)dx = 2\pi$ for all $n \in \N$, the Fourier transform $\hat{\eta}_n \in L^2(\R^2)$ is real-valued, $\hat{\eta}_n(k) = \hat{\eta}(k/n)$, $|\hat{\eta}_n(k)| \leq 1$ and $\hat\eta_n(k) \to 1$ as $n \to \infty$ for all $k \in \R^2$. The regularized quadratic form describing the energy of $N$ fermions with one impurity particle of mass $M$ in $\R^2$ is defined on $L^2(\R^2) \otimes \HH_N$, where  $\HH_N := \bigwedge\nolimits^N L^2(\R^2)$, and given by
\begin{align}
 & \int d\boldsymbol{x} dy \left( \frac{1}{M} |\nabla_y \psi(y, \boldsymbol{x})|^2 + \sum_{i=1}^N |\nabla_{x_i} \psi(y, \boldsymbol{x})|^2 \right) \nonumber \\
 &\qquad\qquad - g_n \sum_{i=1}^N \int dx_1 \, ... \, \widehat{dx_i} \, ... \, dx_N dy \left| \int \! dx_i \, \eta_n(x_i - y) \psi(y, \boldsymbol{x}) \right|^2, \label{Quadratic_Form_R2}
\end{align}
where $\boldsymbol{x}=(x_1,\ldots,x_N)$ and the coupling constant is defined by the renormalization condition
$$
   g_n^{-1} = \int \! dk \: \frac{\hat{\eta}_n(k)^2}{(1+\frac{1}{M}) k^2 - E_B}.
$$
As in the previous sections, $E_B < 0$ is a fixed parameter of the system, which plays the role of the impurity-fermion binding energy (cf. \eqref{Renormalization_Condition}). We now write the quadratic form in terms of center-of-mass and relative coordinates,
$$
   R = \frac{My + \sum_{i=1}^N x_i}{M + N} \qquad \textrm{and} \qquad r_i = x_i - y,
$$
and then make a Fourier transform with respect to all the new coordinates $R,r_1,\ldots,r_N$. The new wave function, resulting from the unitary transformations given by the change of coordinates and the Fourier transform, will  simply be denoted by $\hat\psi$.  The quadratic form now reads
\begin{align}\nonumber
   & \int d\mathbf{k} dP \left( \frac{P^2}{M+N}  + \frac{1}{M} \left|\sum_{i=1}^N k_i \right|^2 + \sum_{i=1}^N k_i^2\right) |\hat{\psi}(P, \mathbf{k})|^2  \\
   &\qquad\qquad - g_n \sum_{i=1}^N \int dk_1\ldots\widehat{dk_i} \ldots dk_N dP \left| \int \! dk_i \, \hat\eta_n(k_i-P/(M+N)) \hat{\psi}(P, \mathbf{k}) \right|^2,\label{FT-form-R2}
\end{align}
and it is defined on the set of all $\hat{\psi}$ with $ \int(P^2+\sum_{i=1}^Nk_i^2)|\hat{\psi}(P, \mathbf{k})|^2\,d\mathbf{k} dP  <\infty$.

Let $H_n$ denote the self-adjoint Hamiltonian associated with this semi-bounded quadratic form. It is unitarily equivalent to the Hamiltonian associated with \eqref{Quadratic_Form_R2}. Then 
\begin{equation}\label{def-Hn-R2}
   H_n = \int_{\R^2}^{\oplus} H_n(P)\, dP,
\end{equation}
where
$$
   H_n(P) = \frac{1}{M+N} P^2 + \Hrel(n,P)
$$
and in the language of second quantization,
$$
      \Hrel(n,P) = H_0-g_n a^*(\hat{\eta}_{n,P})a(\hat{\eta}_{n,P})
$$
with $\hat{\eta}_{n,P}(k) = \hat{\eta}_{n}(k-P(M+N)^{-1})$ and 
$$
    H_0 = \frac{1}{M}P_f^2 + H_f, \qquad P_f :=  \int k a_k^* a_k\, dk, \quad H_f := \int k^2 a_k^* a_k\, dk.
$$
The $P$-dependence of the Hamiltonian $\Hrel(n,P)$ in the center-of-mass frame is somewhat surprising, because non-relativistic many-particle Hamiltonians with  two-body potential are independent of the momentum of the center of mass, after the kinetic energy of the center of mass has been subtracted. The regularized interaction given in \eqref{Quadratic_Form_R2} is not a sum of two-body-potentials, but it approximates a sum of $\delta$-potentials and hence the $P$-dependence of $\Hrel(n,P)$ should disappear as $n\to\infty$. It will disappear as we will see below.

The Hamiltonian $\Hrel(n,P)$ has the general form required for an application of \Cref{Convergence_Theorem}. The verification of the hypotheses of that theorem follows the line of arguments given in Section \ref{Sect:Regularization_Schemes} and the details can be found in \cite{Diss-Linden}. Here we merely summarize the main steps and results:

\begin{enumerate}
\item For all $z\in \rho(H_0)$ the limit 
$$
     B_z := \lim_{n\to\infty} a(\hat{\eta}_{n,P})(H_0-z)^{-1}
$$
exists and it is independent of $P\in\R^2$ and the choice of $\eta$. In particular, $A\ph =  \lim_{n\to\infty} a(\hat{\eta}_{n,P})\ph$ exists for $\ph\in D(H_0)$ and
\begin{equation}\label{eq-for-A}
       (A\ph)(k_1,\ldots,k_n) = \sqrt{N}\int dk\, \ph(k,k_1,\ldots,k_{N-1}).
\end{equation}
\item Let $D=D(H_f)\cap \HH_{N-1}$ and for $z\in \rho(H_0)$ let
$$
    \phi_n(z) = g_n^{-1} -  a(\hat{\eta}_{n,P}) (H_0-z)^{-1} a^{*}(\hat{\eta}_{n,P}).
$$
Then for all $w\in D$ and all $E<0$,
$$
     \phi_n(E)w\to  \phi(E)w,\qquad (n\to\infty)
$$
where $ \phi(E) =  \phi^0(E)+ \phi^I(E)$ is essentially self-adjoint on $D$ and 
\begin{align*}
       \phi^0(E)&= \frac{\pi}{1 + \frac{1}{M}} \log \left( \frac{\frac{1}{M+1} P_f^2 + H_f - E}{-E_B} \right)\\
        \phi^I(E) &= \int \! dk \, dl \: a_k^* \frac{1}{\frac{1}{M}(P_f + k + l)^2 + H_f + k^2 + l^2 - E} \: a_l.
\end{align*}
The operator $\phi^{I}(E)$ is bounded and hence $\phi(E)$ is self-adjoint on $D(\phi^{0}(E))$.
\item For every $c>0$ there exists $E_c\leq -1$ such that for all $\mu\leq E_c$ and all $n\in\N$ sufficiently large,
$$
          \phi_n(\mu) \geq c.
$$
\end{enumerate}
Thanks to \Cref{Convergence_Theorem} there exists a self-adjoint operator $\Hrel$ in $\HH_N$ such that $\Hrel(n,P)\to \Hrel$ in the strong resolvent sense and for $\mu\leq E_{c=1}$,
$$
   (\Hrel - \mu)^{-1} = (H_0-\mu)^{-1} + B_\mu^* \phi(\mu)^{-1} B_\mu.
$$
Since $B_\mu$ and $\phi(\mu)$ are independent of $P$ and $\eta$, so is the operator $\Hrel$. Moreover, by \Cref{lower-bound}, for every $E<0$,
$$
    \phi(E)\geq 0\quad\Rightarrow\quad \Hrel\geq E,
$$
which was announced and used in \cite{GriesemerLinden1}.

\begin{thm}
There exists a semibounded self-adjoint operator $H$ in $L^2(\R^2,\HH_N)$ such that $H_n\to H$ in the strong resolvent sense, where $H_n$ is defined by    \eqref{def-Hn-R2}. If $E<0$ and if the operator $\phi(E)$ in $\HH_{N-1}$ is defined as above, then 
$$
      \phi(E)\geq 0\quad\Rightarrow\quad H\geq E.
$$
\end{thm}

\begin{proof}
Let $H(P):D(\Hrel)\subset \HH_N\to\HH_N$ be defined by 
$$
       H(P) = \frac{P^2}{M+N} + \Hrel,
$$
and let 
$$
      H = \int_{\R^2}^{\oplus} H(P)\, dP.
$$
This means $(H\psi)(P) = H(P)\psi(P)$ where $D(H)$ is the space of all $\psi\in L^2(\R^2,\HH_N)$ such that $\psi(P)\in D(\Hrel)$ for a.e. $P\in \R^2$ and 
$\int \|H(P)\psi(P)\|^2\, dP <\infty$. Then $H$ is self-adjoint and bounded from below because $\Hrel$ is bounded from below. For all $\psi\in L^2(\R^2,\HH_N)$
\begin{align*}
    &\|((H-i)^{-1}-(H_n-i)^{-1})\psi\|^2\\
    &= \int  \|((H(P)-i)^{-1}-(H_n(P)-i)^{-1})\psi(P)\|^2   \, dP\\
    &=\int  \|((\Hrel+P^2/(M+N)-i)^{-1}\\ &\qquad\qquad -(\Hrel(n,P)+P^2/(M+N)-i)^{-1})\psi(P)\|^2   \, dP\\
    &\to 0,\qquad (n\to\infty)
\end{align*}
by Lebesgue dominated convergence and the strong resolvent convergence $\Hrel(n,P)\to \Hrel$.
\end{proof}

We now show that $\Hrel$ is a TMS Hamiltonian associated with the $(N+1)$-particle system in the center-of-mass frame with point interaction among the fermions and the impurity. To this end we set 
$$
       \alpha = -\frac{\pi}{1+M^{-1}}\log(-E_B)
$$
and we define the functions
\begin{align*}
     L_{\lambda}(k_1,\ldots,k_{N-1}) &= \frac{\pi}{1+M^{-1}}\log\left(\frac{1}{M+1}\Bigg|\sum_{j=1}^{N-1}k_j\Bigg|^2 + \sum_{j=1}^{N-1}k_j^2 - \lambda\right)\\
    G_{\lambda}(k_1,\ldots,k_{N}) &= \left(\frac{1}{M}\Bigg|\sum_{j=1}^{N}k_j\Bigg|^2 + \sum_{j=1}^{N}k_j^2 - \lambda\right)^{-1}.
\end{align*}
By definition of $\phi(\lambda)$, for $w\in D(\phi) = D(\phi^0(\lambda))$ and $\lambda<0$,
\begin{align}
   (\phi(\lambda)w)(k_1,\ldots,k_{N-1}) &=\ (\alpha+ L_{\lambda}(k_1,\ldots,k_{N-1}) )w(k_1,\ldots,k_{N-1})\nonumber\\
     &\qquad +\sum_{j=2}^{N}(-1)^{j}\int (G_{\lambda}w_j)(q,k_1,\ldots,k_{N-1})\, dq,\label{eq-for-phi}
\end{align}
and 
\begin{equation}\label{eq-for-Rstar}
    (B_\lambda^{*}w)(k_1,\ldots,k_{N}) =  \frac{1}{\sqrt{N}}\sum_{j=1}^N (-1)^{j+1}(G_{\lambda}w_j)(k_1,\ldots,k_{N})
\end{equation}
where $w_j$ denotes the function defined by
$$
      w_j(k_1,\ldots,k_{N}) = w(k_1,\ldots,k_{j-1},k_{j+1}\ldots, k_{N}).
$$
By \Cref{Explicit_Characterization_of_H}, for given $\lambda<0$ the statement  $\ph\in D(\Hrel)$ is equivalent to the existence of some $w\in D(\phi)$ such that $\ph-B_\lambda^{*}w \in D(H_0)$ and
\begin{equation}\label{TMS-1}
  A(\ph-B_\lambda^{*}w) = \phi(\lambda)w.
\end{equation}
In view of \eqref{eq-for-A} this means, for $\xi=w/\sqrt{N}$, that 
\begin{equation}\label{TMS-2}
   \int (\ph-\sqrt{N}B_\lambda^{*}\xi)(q,k_1,\ldots,k_{N-1})dq  = (\phi(\lambda)\xi)(k_1,\ldots,k_{N-1}),
\end{equation}
for almost all $(k_1,\ldots,k_{N-1})\in \R^{2(N-1)}$. By \eqref{eq-for-Rstar} the left side of this equation becomes
\begin{align*}
    \int \big(\ph(q,k_1,\ldots,k_{N-1}) &- G_\lambda(q,k_1,\ldots,k_{N-1})\xi(k_1,\ldots,k_{N-1})\big)dq\\  &+  \sum_{j=2}^{N}(-1)^{j}\int     (G_{\lambda}\xi_j)(q,k_1,\ldots,k_{N-1})\, dq
\end{align*}
where the second term agrees with the second term of \eqref{eq-for-phi} with $w\to\xi$. Hence \eqref{TMS-2} reduces to
$$
     \int \big(\ph(q,K) - G_\lambda(q,K)\xi(K)\big)\,dq  =  (\alpha+ L_{\lambda}(K))\xi(K),
$$  
where $K=(k_1,\ldots,k_{N-1})$. This means that, in the limit $R\to\infty$,
\begin{align*}
   \int_{|q|\leq R} \ph(q,K)\, dq &= \left(\int_{|q|\leq R} G_\lambda(q,K)\,dq +  \alpha+ L_{\lambda}(K)\right)\xi(K) +o(1)\\
   &= \frac{\pi}{1+M^{-1}} \log \Big(\frac{(1 + M^{-1})R^2}{-E_B}\Big)\xi(K) +o(1).
\end{align*}
This is the TMS boundary condition written in Fourier space \cite{Figari, Five_Italians}.
In the case $N=1$ we see that $\ph\in D(\Hrel)$ if and only there exists some $\xi\in \C$ such that  
$$
        \int(\ph(k) - G_\lambda(k)\xi)\,dk =  \frac{\pi}{1 + M^{-1}} \log(\lambda/E_B)\xi,
$$
which means that
$$
    \int_{|k|\leq R} \ph(k)\,dk  =  \frac{\pi}{1 + M^{-1}} \log \Big(\frac{(1 + M^{-1})R^2}{-E_B}\Big)\xi + o(1),\quad (R\to\infty). 
$$


\section{Two species of fermions}
\label{sec:two-species}

	In this section we explain how the setup and the results of Section~\ref{Sect:Fermi_Polaron_Box} are generalized to systems of two species of fermions. We show that the abstract theory developed in Sections \ref{sec:BS-operator} and \ref{Sect:Strong_Resolvent_Limit} applies to this more general class of systems. 

We consider a system composed of $N_1$ fermions of mass $m_1$ and $N_2$ fermions of mass $m_2$ in a box $\Omega=[0,L]^2\subset\R^2$ with periodic boundary conditions. The Hilbert space of this system is
$$
       \HH_{N_1,N_2} :=  \bigwedge\nolimits^{\!N_1} L^2(\Omega) \otimes \bigwedge\nolimits^{\!N_2} L^2(\Omega)
$$
and the regularized Hamiltonian, in second quantized representation, reads
$$
   H_{\alpha,\beta}:= H_0 - g_{\alpha,\beta} W_{\alpha,\beta}
$$
where $H_0$, $g_{\alpha,\beta}$ and $W_{\alpha,\beta}$ are defined as in Section~\ref{Sect:Fermi_Polaron_Box} with the substitutions 
$M\to m_1$ and $1\to m_2$ of the masses. Explicitly,
\begin{align*}
       H_0 &:= \sum_k k^2 (\tfrac{1}{m_1} a^*_k a_k + \tfrac{1}{m_2} b^*_k b_k),\\
       g_{\alpha,\beta}^{-1} &:= \sum_{k} \frac{\alpha(k)^2 \beta(-k)^2}{(\frac{1}{m_1}+\frac{1}{m_2}) k^2 - E_B}.
\end{align*}
Formally, the operator $V_{\alpha,\beta}$ reads as in  Section~\ref{Sect:Fermi_Polaron_Box} as well, but now
$$
        V_{\alpha,\beta}:  \HH_{N_1,N_2} \to   L^2(\Omega)\otimes \HH_{N_1-1,N_2-1},
$$
where 
$$
    m_q^*:\HH_{N_1,N_2} \to  L^2(\Omega)\otimes \HH_{N_1,N_2}     
$$
creates a state $\ph_q$ in the attached  $L^2(\Omega)$, and $m_p\ph_q = \sprod{\ph_p}{\ph_q}=\delta_{p,q}$. Assumption \eqref{First_Condition_Regularization} that $C(\alpha,\beta)<\infty$ implies that $V_{\alpha,\beta}$ and $W_{\alpha,\beta}$ are bounded operators with 
$$
          \|V_{\alpha,\beta}\| \leq \sqrt{N_1N_2} C(\alpha,\beta)^{1/2}\quad \text{and} \quad W_{\alpha,\beta}=V_{\alpha,\beta}^{*}V_{\alpha,\beta},
$$
generalizing \Cref{Regularized_Interaction}. Again, as in Section~\ref{Sect:Fermi_Polaron_Box}, the parameter $E_B$, by definition of $g_{\alpha,\beta}$, agrees with the ground state energy of $H_{\alpha,\beta}\rst \HH_{1,1}$ in the sector of total momentum $q=0$. This shows, in particular, that the abstract theory described in Section \ref{sec:BS-operator} and \ref{Sect:Strong_Resolvent_Limit} applies to the system of $N_1+N_2$ fermions. We expect the results of Section~\ref{Sect:Regularization_Schemes} to generalize in a straightforward way.

\bigskip\noindent
\textbf{Acknowledgement.} We thank Rodolfo Figari,  Rafaele Carlone, Alessandro Teta, and Jan Philip Solovej for extended discussions and for the hospitality at the university of Naples, the university La Sapienza in Rome, and the university of Copenhagen. We thank Hans Peter B\"uchler for bringing the Fermi polaron to our attention and for several inspiring discussions. We thank an annonymous referee for pointing out the work of Andrea Posilicano and we thank Andrea for helpful remarks about his pertinent publications. Our work was supported by the \emph{Deutsche Forschungsgemeinschaft (DFG)} through the Research Training Group 1838: \emph{Spectral Theory and Dynamics of Quantum Systems}. 


\end{document}